\documentclass[a4paper]{article} 
\usepackage[top=2cm, bottom=2cm, left=2cm, right=2cm]{geometry}
\usepackage[dvipsnames,svgnames]{xcolor}
\usepackage{tikz,pgf}
\usepackage{amsfonts}
\usepackage{amsmath}
\usepackage{mathtools}
\usepackage{amssymb}
\usepackage{amsthm}
\usepackage{mathrsfs}
\usepackage{graphicx}
\usepackage{color}
\usepackage{pstricks}
\usepackage{hyperref}
\usepackage{caption}
\usepackage{subcaption}
\usepackage{epsfig}
\usepackage{epstopdf}
\usepackage{booktabs}
\newcommand\eq[1] {(\ref{#1})}

\newcommand\fig[1] {\ref{fig:#1}}

\newcommand\labfig[1] {\label{fig:#1}}

\newtheorem{theorem}{Theorem}[section]
\newtheorem{corollary}{Corollary}[section]

\newcommand{\bfm}[1]{\mbox{\boldmath ${#1}$}}
\newcommand{\nonum}{\nonumber \\}
\newcommand{\beqa}{\begin{eqnarray}}
\newcommand{\eeqa}[1]{\label{#1}\end{eqnarray}}
\newcommand{\beq}{\begin{equation}}
\newcommand{\eeq}[1]{\label{#1}\end{equation}}
\newcommand{\bpm}{\begin{pmatrix}}
	\newcommand{\epm}{\end{pmatrix}}
\newcommand{\Grad}{\nabla}
\newcommand{\Div}{\nabla \cdot}

\newcommand{\ssigma}{{\mathfrak S}}

\newcommand{\Gd}{\delta}

\newcommand{\Gf}{\phi}

\newcommand{\Gl}{\lambda}

\newcommand{\BGf}{\bfm\phi}

\newcommand{\CA}{{\cal A}}
\newcommand{\CB}{{\cal B}}
\newcommand{\CC}{{\cal C}}

\newcommand{\CF}{{\cal F}}
\newcommand{\CG}{{\cal G}}

\newcommand{\CM}{{\cal M}}

\newcommand{\CR}{{\cal R}}

\def\Ba{{\bf a}}
\def\Bb{{\bf b}}

\def\Be{{\bf e}}
\def\Bf{{\bf f}}

\def\Bm{{\bf m}}
\def\Bn{{\bf n}}

\def\Bu{{\bf u}}
\def\Bv{{\bf v}}
\def\Bw{{\bf w}}
\def\Bx{{\bf x}}
\def\By{{\bf y}}

\def\BA{{\bf A}}

\def\BF{{\bf F}}
\def\BG{{\bf G}}

\def\BI{{\bf I}}

\def\BR{{\bf R}}

\def\BU{{\bf U}}
\def\BV{{\bf V}}
\def\BW{{\bf W}}
\def\BX{{\bf X}}

\def\B0{{\bf 0}}

\def \ba {\begin{array}}
	\def \ea {\end{array}}
\def\restrict#1{\raise-.5ex\hbox{\ensuremath|}_{#1}}
\newtheorem {Thm} {Theorem} [section]

\newtheorem {Adef} [Thm] {Definition}

\newtheorem {Arem} [Thm] {Remark}

\newtheorem {Aexa} [Thm] {Example}

\newtheorem {Anot} [Thm] {Notation}

\def \refe #1.{(\ref{#1})}
\def \reff #1.{figure~\ref{#1}}
\def \refs #1.{section~\ref{#1}}
\def \refss #1.{subsection~\ref{#1}}
\def \refD #1.{Definition~\ref{#1}}
\def \refT #1.{Theorem~\ref{#1}}
\def \refL #1.{Lemma~\ref{#1}}
\def \refC #1.{Corollary~\ref{#1}}
\def \refP #1.{Proposition~\ref{#1}}
\def \refR #1.{Remark~\ref{#1}}
\def \refE #1.{Example~\ref{#1}}
\def \refN #1.{Notation~\ref{#1}}
\usepackage{fancyhdr}

\newtheorem{lemma}{Lemma}

\def\R{{\mathbb R}}

\def\e{{\varepsilon}}

\newcommand{\detail}[1] {}
\makeatletter
\def\divergence{\mathop{\operator@font div}\nolimits}
\makeatother

\title{{On the forces that cable webs under tension can support and how to design cable webs to channel stresses}}

\author{Guy Bouchitt\'e\thanks{Universit\'e de Toulon, e-mail: bouchitte@univ-tln.fr}\;\;, \;
	Ornella Mattei\thanks{Department of Mathematics, University of Utah, e-mail: mattei@math.utah.edu}\;\;, \;
	Graeme W. Milton\thanks{Department of Mathematics, University of Utah, e-mail: milton@math.utah.edu}\;\;, \; and
	Pierre Seppecher\thanks{Universit\'e de Toulon, e-mail: seppecher@imath.fr}}

\begin{document}
	
	\maketitle
	\begin{abstract}
		\noindent
		In many applications of Structural Engineering the following question arises: given a set of forces $\mathbf{f}_1,\mathbf{f}_2,\dots,\mathbf{f}_N$ applied at prescribed points
		$\mathbf{x}_1,\mathbf{x}_2,\dots,\mathbf{x}_N$, under what constraints on the forces does there exist a truss structure (or wire web) with all
		elements under tension that supports these forces? Here we provide answer to such a question for any configuration of the terminal points $\mathbf{x}_1,\mathbf{x}_2,\dots,\mathbf{x}_N$ in the two- and three-dimensional case. Specifically, the existence of a
		web is guaranteed by a necessary and sufficient condition on the loading which corresponds to a finite dimensional linear programming problem. In two-dimensions we show that any such web can be replaced by one in which there are at most $P$ elementary loops, where elementary means
		that the loop cannot be subdivided into subloops, and where $P$ is the number of forces $\mathbf{f}_1,\mathbf{f}_2,\dots,\mathbf{f}_N$ applied at points strictly within the convex hull of 
		$\mathbf{x}_1,\mathbf{x}_2,\dots,\mathbf{x}_N$. In three-dimensions we show that, by slightly perturbing
		$\mathbf{f}_1,\mathbf{f}_2,\dots,\mathbf{f}_N$, there exists a uniloadable web supporting this loading. Uniloadable means it supports this loading
		and all positive multiples of it, but not any other loading. Uniloadable webs provide a mechanism for distributing stress in desired ways.
	\end{abstract}
	
	
	\section{Introduction}\label{Introduction}
	
	One of the main goals of Structural Engineering is to find performing structures when one incorporates into the design a specific type of material or substructure. Many materials behave quite differently under tension or compression: concrete and masonry structures are two examples of materials that perform much better under compression. Some types of structures support also only specific loadings: a wire or cable, for example, can only support tension and not compression. Here we are interested in the case where one incorporates a material that works particularly well under tension so that a cable web is expected to be representative of the most performing structure to be used.
	Thus, we address the following problem: given a set of forces $\Bf_1,\Bf_2,\ldots,\Bf_N$ applied at prescribed points
	$\Bx_1,\Bx_2, \ldots, \Bx_N$, under what constraints on the forces does there exist a wire web (or truss structure) with all elements under tension that supports these forces? Note that the problem is identical if one is interested in the case where all elements are under compression.
	We are only interested in cable wires which can be modeled as a set of straight truss elements: we do not consider the case of catenary elements, see, e.g., \cite{Ahmad:2013:NLA,Ahmadizadeh:2013:TDG,Andreu:2006:AND,Such:2009:AAB}.

In the two-dimensional case, a complete answer to this
problem is given by Theorem 1 in \cite{Milton:2017:SFI} in the special case where the prescribed points are vertices of a convex polygon.
This theorem states that, in this case, a web exists if and only if the net torque going clockwise around any connected portion of the boundary is positive: for any sequence $(\Bx_i, \Bx_{i+1},\dots, \Bx_j)$ of consecutive vertices ordered clockwise (where ${\bf x}_k$ is identified with ${\bf x}_{k-N}$) we have:
\beq  \sum_{k=i}^{j} \det(\Bx_k-\Bx_i, \Bf_k)\geq 0.\eeq{I1}
Furthermore, by using the Airy stress functions theory, a representative web is explicitly given that contains no closed loops (that is, there is no set of wires forming the boundary of a polygon).
The reader is referred to \cite{Milton:2017:SFI} for details.

So the question now is: what happens when the points are not the vertices of a convex polygon or what happens in the three-dimensional case? Theorem 1.1 (see below), which is one of the main results of this paper, answers this question completely. To make the statement clear, let us first introduce some terminology that will be used throughout the paper:\\
\\
\noindent
$\bullet$ {\it Finite web}: a collection $W:=\big([\Bx_i,\Bx_j],[\Bx_k,\Bx_l]\dots \big)$ of segments (or bars) where $\Bx_1,\dots, \Bx_M$ are a finite set of points called {\it nodes}.\\
\noindent
$\bullet$ {\it Terminal nodes}: the nodes $\BX=(\Bx_1,\dots, \Bx_N)$, $N\leq M$, where the forces are applied.\\
\noindent
$\bullet$ {\it Internal nodes}: the remaining nodes, if any. \\
\noindent
$\bullet$ {\it Admissible web stress state}: when each bar $[\Bx_i,\Bx_j]$ of the web is endowed with a non negative tension $\sigma_{ij}$, we say that $\sigma=(\sigma_{ij},\sigma_{kl},\dots ) $ is an admissible web stress state on $W$ for the  loading $\BF$ at $\BX$ if it is in equilibrium under the action of the forces $\BF=(\Bf_1,\Bf_2,\ldots,\Bf_N)$ applied at $\BX=(\Bx_1,\Bx_2,\ldots,\Bx_N)$.
 If there exists such an admissible stress state then the web $W$ is said to {\it support}\footnote{By ``admissible stress state'' we mean an equilibrium state in which all bars are either in tension, or carrying no load. Thus, by ``supporting'', we mean supporting with all bars either in tension, or carrying no load.}  $\BF$ at $\BX$.\\
 \noindent
 $\bullet$ {\it Uniloadable webs}:  webs which support only one loading (up to a positive multiplicative constant).\\
 
Theorem 1.1 then reads as follows: 

\begin{theorem}\label{mainthm1}{\bf Existence of a web under tension}
	\newline
	Let  $\CA_{\BX}$ be the cone of displacements $\BU=(\Bu_1,\Bu_2,\ldots,\Bu_N)$ at points $\BX=(\Bx_1,\Bx_2,\ldots,\Bx_N)$ defined by
	\beq {\CA_{\BX}}:=\{\BU\in (\R^d)^N\ :\ \forall\,  1\le i<j\le N,\ (\Bu_i-\Bu_j)\cdot(\Bx_i-\Bx_j)\geq 0 \}. \eeq{W10}
	Then, the following condition
	\beq \inf_{\BU\in{\CA_{\BX}}}\BF\cdot\BU\geq 0 \eeq{W9}
	is necessary and sufficient to ensure the existence of a finite web under tension that supports the loading $\BF$ at points $\BX$. In such a case, the web connecting the terminal points $\BX$ pairwise supports the loading $\BF$.
\end{theorem}

Section \ref{Admissible_loadings} is dedicated to the proof of Theorem \ref{mainthm1} and related consequences, whereas Section \ref{Mechanical_inter} provides insights on the mechanical meaning of the theorem with special reference to the two-dimensional case. In general, from a mechanical point of view, the statement
says that the work performed by $\BF$ is non negative for any (infinitesimal)\footnote{Notice that we solve this problem within the context of infinitesimal elasticity: examples of applications of the finite deformation theory to describe the geometric nonlinearity are given, for instance, by \cite{Crusells:2017:AMF,Yang:2007:GNA,Thai:2011:NSD}.} displacements $\BU$ corresponding to a global expansion of the system of points $\BX$. Notice that condition \eqref{W9} provides a characterization of the set 
$\CA_{\BX}^*$ of all the loadings $\BF$ at points $\BX$ which can be supported by some finite web as the solution to a finite dimensional linear programming
problem. Moreover, Theorem \ref{mainthm1} states that a web supporting the given loading is the one that connects the terminal points pairwise (see Example \ref{Ex1}).

\begin{Aexa}\label{Ex1}
	Consider a balanced set of forces $\Bf_1$,\ldots,$\Bf_N$ at points $\Bx_1$,\ldots,$\Bx_N$ that are directed radially outwards from a central point $\Bx_0$ (so that $\Bf_i=c_i (\Bx_i-\Bx_0)$ for some set of positive coefficients $c_i$).  Note that the balance of forces implies that $\Bx_0$ must belong to the convex hull of the points  $\Bx_1$,\ldots,$\Bx_N$.
	
	\begin{figure}[!ht]
		\centering
		{	\begin{tikzpicture}[scale=0.50]
			\draw[black, thick] (0,0) -- (-4,0);
			\draw[black, thick] (0,0) -- (-3,4);
			\draw[black, thick] (0,0) -- (2,4);
			\draw[black, thick] (0,0) -- (4,3);
			\draw[black, thick] (0,0) -- (4,-3);
			\filldraw[black] (0,0) circle (2pt);
			\node (X0) at (-0.2,-0.6) {$\Bx_0$}; 
			\filldraw[black] (-4,0) circle (2pt) ;
			\node (X1) at (-4.2,-0.6) {$\Bx_1$};
			\filldraw[black] (-3,4) circle (2pt);
			\node (X2) at (-2.7,4.5) {$\Bx_2$};
			\filldraw[black] (2,4) circle (2pt);
			\node (X3) at (1.5,4.4) {$\Bx_3$};
			\filldraw[black] (4,3) circle (2pt);
			\node (X4) at (4.6,2.6) {$\Bx_4$};
			\filldraw[black] (4,-3) circle (2pt);
			\node (X5) at (4.7,-2.8) {$\Bx_5$};
			\draw[->,red,thick](-4,0)--(-5,0);
			\draw[->,red,thick](-3,4)--(-3.75,5);
			\draw[->,red,thick](2,4)--(2.2,4.4);
			\draw[->,red,thick](4,3)--(4.4,3.3);
			\draw[->,red,thick](4,-3)--(6,-4.5);
			\node (F1) at (-5.3,0){\red $\Bf_1$};
			\node (F2) at (-3.9,5.4){\red $\Bf_2$};
			\node (F3) at (2.3,4.7){\red $\Bf_3$};
			\node (F4) at (4.6,3.6){\red $\Bf_4$};
			\node (F5) at (6.4,-4.1){\red $\Bf_5$};
			\node (a) at (0,-4.5){(a)};
			\end{tikzpicture}}
		{   \begin{tikzpicture}[scale=0.50]
			\draw[black, thick] (-4,0)--(-3,4);
			\draw[black, thick] (-4,0)--(2,4);
			\draw[black, thick] (-4,0)--(4,3);
			\draw[black, thick] (-4,0)--(4,-3);
			\draw[black, thick] (-3,4)--(2,4);
			\draw[black, thick] (-3,4)--(4,3);
			\draw[black, thick] (-3,4)--(4,-3);
			\draw[black, thick] (2,4)--(4,-3);
			\draw[black, thick] (4,3)--(4,-3);
			\draw[black, thick] (4,3)--(2,4);
			\filldraw[black] (0,0) circle (2pt);
			\node (X0) at (-0.2,-0.6) {$\Bx_0$}; 
			\filldraw[black] (-4,0) circle (2pt) ;
			\node (X1) at (-4.2,-0.6) {$\Bx_1$};
			\filldraw[black] (-3,4) circle (2pt);
			\node (X2) at (-2.7,4.5) {$\Bx_2$};
			\filldraw[black] (2,4) circle (2pt);
			\node (X3) at (1.5,4.4) {$\Bx_3$};
			\filldraw[black] (4,3) circle (2pt);
			\node (X4) at (4.6,2.6) {$\Bx_4$};
			\filldraw[black] (4,-3) circle (2pt);
			\node (X5) at (4.7,-2.8) {$\Bx_5$};
			\draw[->,red,thick](-4,0)--(-5,0);
			\draw[->,red,thick](-3,4)--(-3.75,5);
			\draw[->,red,thick](2,4)--(2.2,4.4);
			\draw[->,red,thick](4,3)--(4.4,3.3);
			\draw[->,red,thick](4,-3)--(6,-4.5);
			\node (F1) at (-5.3,0){\red $\Bf_1$};
			\node (F2) at (-3.9,5.4){\red $\Bf_2$};
			\node (F3) at (2.3,4.7){\red $\Bf_3$};
			\node (F4) at (4.6,3.6){\red $\Bf_4$};
			\node (F5) at (6.4,-4.1){\red $\Bf_5$};
			\node (b) at (0,-4.5){(b)};
			\end{tikzpicture}}
		\caption{In this example the forces $\BF$ are directed radially outwards from a central point $\Bx_0$ and so the web that connects the terminal points $\BX$ to $\Bx_0$ supports such a loading. Among all the webs that can support $\BF$, Theorem \ref{mainthm1} provides the web connecting the terminal points pairwise and in Example \ref{Ex1} we determine the stress state in each wire of such a web and we prove it is an equilibrium stress state. }
		\labfig{1}
	\end{figure}
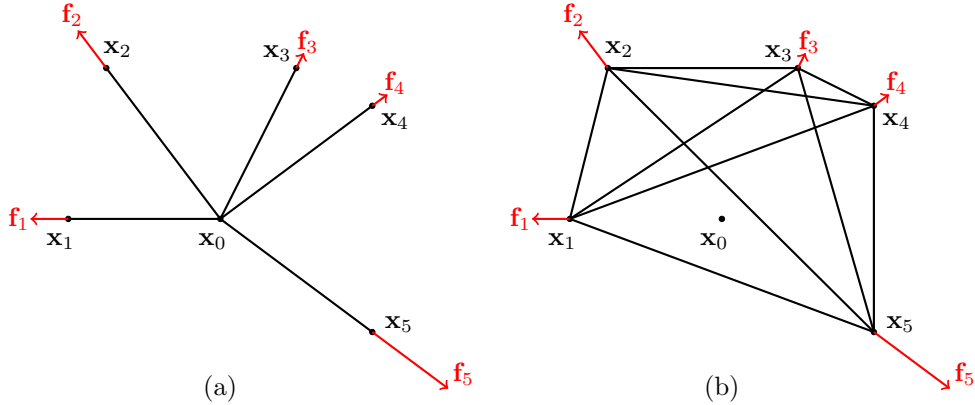
	
	Clearly, the web formed by the wires connecting the points $\Bx_1$,\ldots,$\Bx_N$ to $\Bx_0$, Figure \ref{fig:1}(a), supports this loading: $\sigma_{i0}:=c_i \|\Bx_i-\Bx_0\|^{-1}$ is an admissible stress state on this web.
	But we can find another web which supports the same loading: indeed, the web that connects the points $\Bx_1$,\ldots,$\Bx_N$ pairwise, Figure \ref{fig:1}(b), is suitable when endowed with the stress state $\sigma_{ij}= \|\Bx_i-\Bx_j\| c_i c_j (\sum_k c_k)^{-1}$, with the equilibrium condition
	\begin{align*}
	\Bf_i + \sum_{j\not=i} \sigma_{ij} \frac {\Bx_j-\Bx_i}{\|\Bx_j-\Bx_i\|} &= c_i \left(\Bx_i-\Bx_0\right)+c_i \left(\sum_k c_k\right)^{-1}  \sum_j  c_j (\Bx_j-\Bx_i) \\
	&= c_i (\Bx_i-\Bx_0)
	+c_i \left(\sum_k c_k\right)^{-1}  \left(\sum_j  \Bf_j-\sum_j  c_j (\Bx_i-\Bx_0) \right) =0\end{align*}
	being satisfied at each node $\Bx_i$.
	
\end{Aexa}

Section \ref{Channelling} focuses on another major topic of the paper, that is, on channeling or redistributing stresses: we are interested in webs able to channel forces in a controlled way.
For example, if one considers, say, a bicycle wheel or a suspension bridge, then, a desired distribution of forces
is usually achieved by appropriately tightening the spokes or cables (clearly, the layout of these substructures is also essential). By contrast here we seek to distribute stress through judicious choices of the geometry of the web.

Notice that distributing stresses in wires is quite
different than distributing electrical currents in conducting wires. At a junction of more than two conducting wires one cannot tell in advance
(without looking at the rest of the circuit) how much flow will be  channeled into the different wires (this is an advantage if one wants the current to flow where most needed, a disadvantage if one wants to control the allocation of current, as in an irrigation system). This is due to the fact that, in a conduction network, one only has to satisfy Kirchoff's law which states that the
outgoing current has to balance the incoming current. By contrast, when distributing stresses, one has that at a node where four noncoplanar wires join (for the three-dimensional case, or three non collinear wires for the two-dimensional case), balance of forces implies that the tension in one wire, and the geometry of the junction, uniquely determines the tension in the other wires. Thus, having at each internal node a coordination number of four for the three-dimensional case, or three for the two-dimensional case, is important to uniquely determining the loading that a web can support. 
This principle underlies the construction of ``pentamode materials'' \cite{Milton:1995:WET}, which have a diamond-like structure with a  coordination number of 4 at each node of the structure. As a consequence, the stress field
is essentially uniquely determined: like fluids which only support a hydrostatic loading, pentamodes only support one loading but, unlike fluids, that loading can be a combination of hydrostatic and shear forces. Pentamode materials have been studied for their use in cloaking, in particular for cloaking against sonar. They can guide the acoustic wave around an object while having little impedance mismatch at the membrane boundary with the surrounding fluid \cite{Norris:2008:ACT}.

Webs of springs that support only one loading (up to a multiplicative constant) were instrumental in \cite{Camar:2003:DCS,Vasquez:2011:CCS}. The corresponding elastic energies have the form $(\BF\cdot\BU)^2$, when their terminal nodes have displacements $\BU=(\Bu_1,\Bu_2,\ldots,\Bu_N)$, so that the webs are able to support only forces proportional to $\BF$. In that context, as the elastic energy can be expressed using a rank-one matrix in the form $\BU\cdot (\BF\otimes\BF)\cdot \BU$ , such webs were called ``rank-one webs''. In this
paper, as we are interested in energies that are not necessarily quadratic, we prefer to use the definition of {\it uniloadable webs}. In Section \ref{Channelling} we establish that, apart from some exceptional cases, if forces $\Bf_1,\Bf_2,\ldots,\Bf_N$ at prescribed points $\Bx_1,\Bx_2, \ldots, \Bx_N$ are supported by some web, they are also supported by a uniloadable one.

The uniloadable networks we introduce here if prestressed could replace the pentamode materials in the aforementioned cloaking applications, allowing much greater flexibility in the design and economy in the number of junctions and wire elements.

One may not be interested only in uniloadable webs but more generally in webs $W$ for which the set $\CC^W_\BX$ of supported loadings is prescribed. It is clear that, for a given web  $W$ with terminal nodes at $\BX$, the set $\CC^W_\BX$  is a convex cone contained in $\CA^*_\BX$.  In Section \ref{Channelling}, we show that the converse
is true too in an approximate sense: given a convex cone $\CC\subset\CA_\BX^*$, one can find a sequence of webs $W_n$
such that $\CC^{W_n}_{\BX}$ approaches $\CC$ as $n\to\infty$. For two-dimensional webs where the points $\BX$ are the vertices of a convex polygon, a similar question was addressed in
Theorem 2 of \cite{Milton:2017:SFI}.

\section{On the existence of a web under tension}\label{Admissible_loadings}

This Section is dedicated to the proof of Theorem \ref{mainthm1} and it contains some mathematical technicalities. Therefore, we recommend the reader who is more interested in the mechanical interpretation of the theorem to skip to Section \ref{Mechanical_inter}.

Given a set of forces $\BF$ applied at the points $\BX$, we question the existence of a web supporting such a loading that has all the wires under tension. Recall that the equilibrium of a wire web is achieved if the tension is constant in each wire and if each node is in equilibrium. This situation admits a nice and synthetic, even if a bit abstract, formulation in terms of measures which is convenient for proving Theorem \ref{mainthm1}.

Given a web $W$, the associated measure is
\begin{equation}{\mathcal W}:={\cal H}^1\restrict {[\Bx_i,\Bx_j]}+{\cal H}^1\restrict {[\Bx_k,\Bx_l]}+\dots \label{WW}\end{equation}
where ${\cal H}^1\restrict {[\Bx_i,\Bx_j]}$ stands for the line measure (the one dimensional Hausdorff measure) concentrated on the segment $[\Bx_i,\Bx_j]$. Accordingly the stress state $\sigma$ can be represented by the measure
\begin{equation} \ssigma:=\sigma_{ij} \frac {\Bx_i-\Bx_j}{\|\Bx_i-\Bx_j\|}\otimes \frac {\Bx_i-\Bx_j}{\|\Bx_i-\Bx_j\|} {\cal H}^1\restrict {[\Bx_i,\Bx_j]}+\sigma_{kl} \frac {\Bx_k-\Bx_l}{\|\Bx_k-\Bx_l\|}\otimes \frac {\Bx_k-\Bx_l}{\|\Bx_k-\Bx_l\|} {\cal H}^1\restrict {[\Bx_k,\Bx_l]}+\dots  \label{ssigma} \end{equation}
Similarly, associated with the applied system of forces is the discrete measure
\begin{equation} {\cal F}:= \sum_{i+1}^N \Bf_i\, \delta_{\Bx_i} \label{FF} \end{equation}
where  $\delta_{\Bx_i}$ stands for the Dirac measure at point $\Bx_i$.  The equilibrium condition then simply reads
\begin{equation}\Div\ssigma+{\cal F}=0 \label{eequi}\end{equation}
and the requirement that all wires be under tension is equivalent to the requirement that the measure $\ssigma$ take values in the set of positive semi-definite symmetric matrices. We will denote with ${\cal M}^+$ the set of such measures.
Using this formulation, searching for a finite web boils down to finding $\ssigma\in {\cal M}^+$ of the form \eqref{ssigma} such that \eqref{eequi} is satisfied.

If we drop the constraint that $\ssigma$ must be of the form \eqref{ssigma}, we are led to an interesting generalization: a {\it generalized web} ${\mathcal W}$ is a positive measure, and it supports the loading ${\mathcal F}$ if there exists some  $\ssigma\in {\cal M}^+$ absolutely continuous with respect to ${\mathcal W}$ satisfying the equilibrium condition \eqref{eequi}.
The relation between this generalized formulation and the so-called Michell problem is provided in Section \ref{Michell}.


\subsection{Proof of Theorem \ref{mainthm1}}

\begin{proof}
	To prove that condition  \eq{W9} is necessary, we consider an admissible web stress  $\ssigma$  for the loading $\BF$.
	As  $\ssigma$ is a positive semidefinite tensor-measure with compact support, by Green's generalized formula
	we have (see Section 6 of \cite{Milton:2017:SFI}):
	\beq 0\leq \Big<\ssigma, \Be(\Bu)\Big> =  - \Big<\Div \ssigma,\Bu\Big>  = \Big<{\cal F}, \Bu\Big> = \BF\cdot \BU.
	\eeq{NS-1}
	for all $C^1$ fields $\Bu$ such that $\Be(\Bu):= (\Grad\Bu(\Bx)+(\Grad\Bu(\Bx))^T)/2$ is positive semidefinite.
	Now, let $\BU=(\Bu_1,\Bu_2,\ldots,\Bu_N)$ be an element of $\CA_\BX$.
	For any $\kappa>0$, Lemma \ref{lemmaext} (see below) provides a Lipschitz extension $\widetilde\Bu: \R^d \to \R^d$ satisfying $\widetilde\Bu(\Bx_\ell)=\Bu_\ell$
	and
	\beq \forall (\Bx,\By)\in (\R^d)^2,\  (\widetilde \Bu(\Bx)-\widetilde \Bu(\By))\cdot(\Bx-\By)\geq -\kappa \|\Bx-\By\|^2. \eeq{LE1dup}
	This extension field is differentiable a.e. and, at every point $\Bx$ of differentiability, inequality \eqref{LE1dup} implies that $\Be(\widetilde \Bu)\geq -\kappa\BI$.
	In order to apply Green's formula, we introduce a regularized field $\Bu^\eta:= \widetilde\Bu * \rho^\eta$, $\rho^\eta$ being a smooth convolution mollifier (i.e., non negative, supported in the ball of radius $\eta$, and such that $\int \rho^\eta =1 $).
	It can be readily checked that the strain  $\Be(\Bu^\eta)$  associated with $\Bu^\eta$ is smooth and satisfies everywhere the inequality  $\Be(\Bu^\eta)\geq -\kappa\BI$.
	By applying \eqref{NS-1} to the field $\Bu= \Bu^\eta + \kappa \Bx$, we deduce that
	$$ \sum_i\Bf_i\cdot\Bu^\eta(\Bx_i)+ \kappa\sum_i\Bf_i\cdot\Bx_i \ge 0.
	$$
	As $\Bu^\eta\to \widetilde \Bu $ uniformly on every compact set, we can pass to the limits $\eta \to 0$ and $\kappa\to 0$ in the last inequality to obtain the desired inequality $ \BF \cdot \BU \ge 0$ (see equation \eqref{W9}).
	
	\medskip
	
	To prove that condition \eqref{W9} is sufficient, we consider $\BF$ such that $\BF \cdot \BU \ge 0$ for any $\Bu\in \CA_\BX$.
	The linear conditions $(\Bu_i-\Bu_j)\cdot(\Bx_i-\Bx_j)\geq 0$ can be rewritten in the form $\BF^{(i,j)}\cdot\BU\geq 0$, where:
	\beq \BF^{(i,j)}=(\Bf_1^{(i,j)},\Bf_2^{(i,j)},\ldots,\Bf_N^{(i,j)})\quad \text{with}\quad \Bf_\ell^{(i,j)} = \begin{cases}  \Bx_j-\Bx_i &\text{  if  }\ell=i, \\
		\Bx_i-\Bx_j &\text{  if  }\ell=j, \\
		0 &\text{  otherwise.}  \end{cases}
	\eeq{RT4}
	Hence, thanks to Farkas Lemma\cite{Farkas:1902:TDE}, we know that $\BF$ is a linear combination of the linear forms $\BF^{(i,j)}$ with non negative coefficients:
	\beq \BF=\sum_{1\le i<j\le N} \lambda_{i,j}\, \BF^{(i,j)} \quad ,\quad  \lambda_{i,j}\ge 0.
	\eeq{RT4a}
	It is then easy to check that the  positive semidefinite symmetric tensor measure
	\beq \ssigma= \sum_{1\le i<j\le N}\lambda_{i,j}\, \frac {(\Bx_i-\Bx_j)\otimes (\Bx_i-\Bx_j)}{\|\Bx_i-\Bx_j\|} {\cal H}^1\restrict {[\Bx_i,\Bx_j]}, \eeq{RT2}
	is a possible web stress for the loading $\cal F$.
	Indeed,
	for any $\Ba$ and $\Bb$ in $\R^d$, we have\linebreak $\Div((\Bb-\Ba)\otimes (\Bb-\Ba) {\cal H}^1\restrict{|[\Ba,\Bb]})=\|\Bb-\Ba\|\,  (\Bb-\Ba) (\delta_\Ba-\delta_\Bb)$, thus
	$ \Div\ssigma=-\sum_i \Bf_i \delta_{\Bx_i}=-{\cal F}.$
	Let us emphasize that this web stress measure involves only the original nodes $(\Bx_1, \Bx_2, \dots, \Bx_N)$ as claimed in the theorem.
\end{proof}

We now state and prove the interpolation lemma which was needed in the previous proof.
\begin{lemma}\label{lemmaext}{\bf Interpolation Lemma}
	\newline
	Let $A$ {be} a finite subset of $\R^d$, {and} $\Bu:\ A\to \R^d$ {the field} satisfying
	$(\Bu(\Bx)-\Bu(\By))\cdot(\Bx-\By)\geq 0, \forall (\Bx,\By)\in A^2. $
	Then, for any $\kappa>0$ there exists a Lipschitz extension $\widetilde \Bu$ of $\Bu$ on $\R^d$ satisfying
	\beq \forall (\Bx,\By)\in (\R^d)^2,\  (\widetilde \Bu(\Bx)-\widetilde \Bu(\By))\cdot(\Bx-\By)\geq -\kappa \|\Bx-\By\|^2. \eeq{LE1}
\end{lemma}
\begin{proof}
	As $\Bu$ is bounded on the bounded set $A$, there exists $M$ such that, for any $\Bx\in A$, $\|\Bu(\Bx)\|\leq M$ and $\|\Bx\|\leq M$. Moreover, as $A$ is finite, there exists $\Gd>0$ such that, for any  distinct points $\Bx$ and $\By$ in $A$,  $\|\Bx-\By\|\geq \Gd$. We set $\lambda=\frac {8\, M^2}{\Gd^2}$ and choose $s$ such that $0<\lambda\, s<\min(\kappa,1) $. Let us consider $\BGf$ defined on $A$ by\footnote{Physically, when $s$ is small, we can think of $\BGf(\Bx)$ as a deformation $\Bx\to\BGf(\Bx)$ associated with the displacement field $\Bu(\Bx)$. The additional uniform contractive factor
		of $(1-\lambda\, s^2)$ is needed to account for the finite deformation corrections when $\Bu(\Bx)-\Bu(\By)$ is perpendicular to $\Bx-\By$ (with say $\Bu(\By)=0$ and $\By=0$,
		and $\Bu(\Bx)$ orthogonal to $\Bx$, the distance $|\Bx-s\Bu(\Bx)|$ lengthens as $s$ is increased).}
	\beq \BGf(\Bx)=(1-\lambda\, s^2) \Bx - s \Bu(\Bx). \eeq{LE2}
	For any  distinct points $\Bx$ and $\By$ in $A$, we have, by straightforward computation
	\beqa
	\|\BGf(\Bx)- \BGf(\By)\|^2&=&\|(1-\lambda\, s^2)(\Bx-\By)-s(\Bu(\Bx)-\Bu(\By))\|^2\nonum
	&=&\left\|(\Bx-\By)-s\left(\Bu(\Bx)-\Bu(\By)+\lambda s (\Bx-\By)\right)\right\|^2\nonum
	&=& \|\Bx-\By\|^2 + s^2\|\Bu(\Bx)-\Bu(\By)+\lambda s(\Bx-\By)\|^2\nonum
	&~&\hskip 5 cm -2 s \big((\Bx-\By)\cdot (\Bu(\Bx)-\Bu(\By)+\lambda s\|\Bx-\By\|^2\big)\nonum
	&\leq& \|\Bx-\By\|^2 + s^2\|\Bu(\Bx)-\Bu(\By)+\lambda s(\Bx-\By)\|^2- 2 \lambda s^2 \|\Bx-\By\|^2\nonum
	&\leq& \|\Bx-\By\|^2 + 4 s^2\left(\|\Bu(\Bx)\|^2+\|\Bu(\By)\|^2+(\lambda s)^2\|\Bx\|^2+(\lambda s)^2\|\By\|^2\right)- 2 \lambda s^2 \|\Bx-\By\|^2\nonum
	&\leq& \|\Bx-\By\|^2 + 16 M^2\, s^2-2 \lambda \, \Gd^2\,s^2 \nonum
	&\leq& \|\Bx-\By\|^2.
	\eeqa{LE3}
	Kirszbraun's Theorem \cite{Kirszbraun:1934:UDZ} proves the existence of an extension $\widetilde \BGf$ of $\BGf$ on $\R^d$ satisfying the same condition: for any $(\Bx,\By)\in(\R^d)^2$,
	\beq \|\widetilde\BGf(\Bx)- \widetilde\BGf(\By)\|\leq \|\Bx-\By\|. \eeq{LE4}
	Let us now define $\widetilde \Bu$ on $\R^d$ by setting
	\beq \widetilde \Bu(\Bx)=\frac{(1-\lambda\, s^2)\Bx-\widetilde\BGf(\Bx)}{s}. \eeq{LE5}
	For any $\Bx\in A$ we have $\widetilde \Bu(\Bx)=\Bu(\Bx)$. Moreover  $\widetilde \Bu$ is a Lipschitz function which satisfies
	\beqa
	(\widetilde \Bu(\Bx)-\widetilde \Bu(\By))\cdot (\Bx-\By)&= &s^{-1}\Big((1-\lambda\, s^2) \|\Bx-\By\|^2-(\widetilde \BGf(\Bx)-\widetilde \BGf(\By)) \cdot (\Bx-\By)\Big)\nonum
	&\geq & -\lambda\, s\, \|\Bx-\By\|^2\geq -\kappa \|\Bx-\By\|^2.
	\eeqa{LE6}
\end{proof}

\subsection{Support of the stress field}

From a physical viewpoint it seems obvious that a finite web supporting the loading $\BF$ at the points $\BX=(\Bx_1,\Bx_2,\ldots,\Bx_N)$ should have an associated stress measure vanishing outside the convex hull of the points $\Bx_1,\Bx_2,\ldots,\Bx_N$. This section is devoted to proving this fact which turns out to be valid also for generalized webs.
\begin{theorem}\label{Thsupport}
	Let $\cal F$ be a vector measure with compact support $K$ and $\ssigma$ in $\CM^+$ such that $\Div\ssigma+{\cal F}=0$. Then the support of $\ssigma$ is contained in the convex envelope $co(K)$ of $K$.
\end{theorem}

Notice that, from this result, we can deduce that the support of $\ssigma$ is contained in the subspace spanned by the vectors\footnote{Here we choose $\Bx_1$ as the origin for identifying points and vectors.}  $\Bx_1$, $\Bx_2$, $\ldots$, $\Bx_N$ and thus that the loading forces $\Bf_i$ belong to this subspace. Hence, we will be able to reduce our problem to this subspace and assume without loss of generality that  
$\Bx_1$, $\Bx_2$, $\ldots$, $\Bx_N$ span $\R^d$.

\medskip

\begin{proof} By Hahn-Banach separation theorem \cite{Ekeland:1999:CAV}, proving that $\ssigma$ vanishes outside  the convex envelope $co(K)$ of $K$ reduces to checking that $\ssigma$ vanishes on all half spaces  $P^+_{\Bm,a}:=\{\Bx:\, \Bx\cdot \Bm>a\}$ with  $a\in\R$ and $\Bm\in S^{d-1}$ which do not intersect $K$. This verification is achieved in several steps. We first remark that, as  $P^+_{\Bm,a}$ and $K$ do not intersect, we have $\Div\ssigma=0$ on  $P^+_{\Bm,a}$. 
	
	{\bf Step 1:} Let us consider first a field  $\widetilde \ssigma\in C^1(\R^d, \R^{d^2}_{sym})$ such that $\widetilde \ssigma\geq 0$ and $\Div\widetilde \ssigma=0$ on $P^+_{\Bm,a}$.  Without loss of generality assume that $\Bm=\Be_1$, where $\Be_1$ is the unit vector along the $x_1$-axis
	(where $x_1$ is not to be confused with the point $\Bx_1$).
	Let us  apply Green's formula in the half ball $\Omega_R:=\{\Bx: x_1>a,\ |\Bx|<R\}$. We have
	\beq \int_{\partial \Omega_R} (\widetilde \ssigma\cdot \Be_1)\cdot \Bn\, {\rm d}{\cal H}^{d-1}=\int_{\Omega_R} (\Div \widetilde \ssigma)\cdot \Be_1 \, {\rm d}{\cal H}^{d}=0, \eeq{S0}
	where $\Bn$ stands for the outward normal. Dividing the boundary $\partial \Omega_R$ into $\Sigma_R:=\{\Bx:\ x_1=a,\ |\Bx|<R\}$ and $S_R^+:=\{\Bx:\ x_1>a,\ |\Bx|=R\}$ we get
	\beq -\int_{\Sigma_R} (\widetilde \ssigma\cdot \Be_1)\cdot \Be_1\, {\rm d}{\cal H}^{d-1} + \int_{S_R^+} (\widetilde \ssigma\cdot \Be_1)\cdot \Bn\, {\rm d}{\cal H}^{d-1}=0. \eeq{S1}
	Thus we have
	\beq \int_{\Sigma_R} (\widetilde \ssigma\cdot \Be_1)\cdot \Be_1\, {\rm d}{\cal H}^{d-1} \leq \int_{S_R^+} |\widetilde \ssigma|\, {\rm d}{\cal H}^{d-1}. \eeq{S2}
	As $\widetilde \ssigma\in L^1$, $\int_0^{+\infty}( \int_{S_R^+} |\widetilde \ssigma|\, {\rm d}{\cal H}^{d-1})\, {\rm d}R   <+\infty$ and so there exists a sequence $R_n\to +\infty$ such that $$\int_{S_{R_n}^+} |\widetilde \ssigma|\, {\rm d}{\cal H}^{d-1}\to 0. $$
	This implies that 
	\beq \lim_{n\to \infty} \int_{\Sigma_{R_n}}  (\widetilde \ssigma\cdot \Be_1)\cdot \Be_1\, {\rm d}{\cal H}^{d-1} =0. \eeq{S3}
	As $(\widetilde \ssigma\cdot \Be_1)\cdot \Be_1\geq 0$ we get $(\widetilde \ssigma\cdot \Be_1)\cdot \Be_1= 0$ on the whole hyperplane $\{\Bx: \ x_1=a\}$. As $\widetilde \ssigma\geq 0$, we deduce that  $\ssigma\cdot \Be_1= 0$ on this hyperplane.
	
	\medskip \noindent {\bf Step 2:} We remark that, if $\widetilde \ssigma\geq 0 \in C^1(\R^d, \R^{d^2}_{sym})$ satisfies  $\Div\widetilde \ssigma={0}$ on the whole space $\R^d$, then by applying the result of step 1 to every pair $(\Bm,a)\in S^{d-1}\times \R^d$ we get $\widetilde \ssigma={0}$ everywhere. We also remark that this result holds true for any space dimension $d$.
	
	\medskip \noindent {\bf Step 3:} Going back to the case where $\widetilde \ssigma\geq 0 \in C^1(\R^d, \R^{d^2}_{sym})$ satisfies  $\Div\widetilde \ssigma={0}$ only on the half space $P^+_{\Bm,a}$, we apply Step 1 to every pair $(\Bm,t)$ with $t>a$ and we deduce that $\widetilde \ssigma\cdot \Bm=0$ on  the hyperplane $\Sigma_t:=\{\Bx:\ \Bx\cdot \Bm=t\}$. The restriction of $\widetilde \ssigma$ to this  hyperplane is tangential and divergence free. For almost every $t>a$, $\int_{\Sigma_t}|\widetilde \ssigma|<+\infty$, by applying the result of Step 2 to $\R^{d-1}$, we deduce that $\widetilde \ssigma$ vanishes on almost all hyperplanes $\Sigma_t$. Hence, $\widetilde \ssigma$ vanishes on $P^+_{\Bm,a}$.
	
	\medskip \noindent {\bf Step 4:} In order to extend this result to measures, we introduce a smooth mollifier $\rho^\eta$ and $\ssigma^\eta := \ssigma \star \rho^\eta$. Clearly $\ssigma^\eta$ belongs to $L^1\cap C^{\infty}(\R^{d^2}_{sym})$, $\ssigma^\eta\geq 0$ and, for any $t>a$, and for $\eta$ small enough we have $\Div\ssigma^\eta={0}$ on $P^+_{\Bm,t}$. By applying the result of Step 3 to  $\ssigma^\eta$ we get $\ssigma^\eta={0}$ on $P^+_{\Bm,t}$, and by passing to the limit $\eta\to 0$ we get $\ssigma={0}$ on $P^+_{\Bm,t}$. The theorem is proven by passing finally to the limit $t\to a$.
\end{proof}

\subsection{Link with the Michell problem}\label{Michell}

A very old problem in Mechanical Engineering consists in minimizing the 
total volume of a network of elastic bars ({\em trusses})
while the resistance to a given load remains constant. It reduces to a linear programming problem  which, according to our notation, reads:
\begin{equation}
\inf_\sigma \left\{ \sum_{i} \sum_j  |\sigma_{ij}| \|\Bx_j-\Bx_i\| 
\ :\     \Bf_i + \sum_{j\not=i} \sigma_{ij} \frac 
{\Bx_j-\Bx_i}{\|\Bx_j-\Bx_i\|} =0 \ ,\  \forall i   \right\}
\end{equation}
The classical dual formulation is the following maximization problem on 
deformations:
\begin{equation}
\sup_{\BU}  \left\{  \sum_i\Bf_i\cdot\Bu_i  \ :\  
|(\Bu_i-\Bu_j)\cdot(\Bx_i-\Bx_j)|\le \|\Bx_i-\Bx_j)\|^2   \ ,\ \forall 
i\not=j   \right\}
\end{equation}
As no assumption is made on the number of bars, this belongs to the class of topological optimization problems and it
is well known that, in general, no optimal solution exists. Indeed, during 
the optimization process, the number of bars may increase to infinity 
leading thus to diffuse structures.
The crucial contribution of Michell \cite{Michell:1904:LEM} in the 1900's was 
to formulate a generalized version
(called {\em Michell problem}) in order to take into account all 
possible structures which may appear in the limit. In the generalized 
version, attention is focused on the stress carried by the structure 
rather than on its geometry. Michell stated a duality principle and 
obtained the optimality conditions on the stress and strain tensors: 
they share the same eigenvectors (principal directions) and the 
eigenvalues of the strain tensor have a fixed absolute value. Moreover, 
Michell noticed that, in the two-dimensional case,
when the eigenvalues of the strain tensor have opposite sign and when 
the eigenvector fields are smooth enough to define stream lines (called 
``{\em lines of principal action}''), then these lines constitute a 
so-called {\em Hencky-net}. This is a family
of orthogonal curves which represents the limit of the families of bars 
through the optimization process.
We refer to \cite{Bouchitte:2008:MTL} for a detailed  mathematical 
study where  optimality conditions for a generalized truss are established in a rigorous way.
The generalized stress formulation derived by Michell in dimension $d=2$ 
reads as follows:
\begin{equation}
\inf_\ssigma \left\{ \int \rho_0(\ssigma) \  :\    \Div \ssigma + \CF 
=0  \right\}
\end{equation}
where $\rho_0(\ssigma):=|\lambda_1(\ssigma)|+|\lambda_2(\ssigma)|$
denotes the sum of the moduli of the two principal values  of $\ssigma$.
The corresponding dual problem reads:
\begin{equation} \sup_{\widetilde{\Bu}}\{<{\cal F},\widetilde{\Bu}>: \ 
|({\widetilde{\Bu}}(\Bx)-{\widetilde{\Bu}}(\By))\cdot (\Bx-\By)|\leq 
|\Bx-\By|^2\}
\end{equation}
An admissible pair $(\Bu,\ssigma)$ is then optimal if and only if the 
following extremality relation holds (see \cite{Bouchitte:2008:MTL})
\beq <{\cal F},\Bu>\ =\ \int \rho_0(\ssigma)\,{\rm d}\Bx. \eeq{Sopti}

We can now emphasize the link between our problem and the 
Michell problem described above.
It turns out indeed that admissible stress states associated with a 
loading $\BF$ in the cone $\CA_{\BX}^*$
are solutions of the Michell problem.


\begin{theorem}
	Let ${\cal F}$ be a bounded vector-measure with compact support $K$ 
	and $\ssigma\geq 0$ in $M^{d^2}_{sym}$ such that $\Div\ssigma+{\cal 
		F}={0}$. Then
	
	\noindent i) $\ssigma$ is a solution to the Michell problem for 
	${\cal F}$,
	
	\noindent ii) any solution $\widetilde \ssigma$ to the Michell problem 
	for ${\cal F}$ satisfies $\widetilde \ssigma \geq 0$.
\end{theorem}

\begin{proof}  Clearly $\Bu(\Bx):=\Bx$ is admissible for the dual problem. 
	As $-\Div\ssigma={\cal F}$,  the following relation holds:
	\beq <{\cal F},\Bx>\ =\ <-\Div\ssigma,\Bx>\ =\ <\ssigma,\nabla \Bx>\ =\ \int \mathrm{Tr}(\ssigma). \eeq{S5}
	If we assume that  $\ssigma\geq 0$, then Michell's dual energy $\rho_0(\ssigma)$ coincides with $\mathrm{Tr}(\ssigma)$ as both $\lambda_1(\ssigma)$ and $\lambda_2(\ssigma)$ are non-negative. In particular, we have $\int \rho_0(\ssigma) =\int \mathrm{Tr}(\ssigma)$, hence  by \eqref{S5} the pair  $(\Bx,\ssigma)$  satisfies relation \eqref{Sopti}. The optimality of  $(\Bx,\ssigma)$  follows. This proves point i).
	
	To prove point (ii), it is enough to notice that, for any other solution $\widetilde \ssigma$ of the primal problem, the couple $(\Bx,\widetilde\ssigma)$  is optimal and thus must satisfy the optimality condition \eqref{Sopti}
	so that, by \eqref{S5}, one has 
	$$ <{\cal F},\Bx>=\int \mathrm{Tr}(\widetilde \ssigma)\,{\rm d}\Bx=\int \rho_0(\widetilde \ssigma)\,{\rm d}\Bx. $$
	As  $\rho_0(\widetilde \ssigma)\geq \mathrm{Tr}(\widetilde \ssigma)$,  we deduce that  $\rho_0(\widetilde \ssigma)= \mathrm{Tr}(\widetilde \ssigma)$ and so $\widetilde \ssigma\geq 0$.
\end{proof}

\section{Theorem \ref{mainthm1} in two dimensions}\label{Mechanical_inter}
To get a better insight on the mechanical interpretation of the statement of Theorem \ref{mainthm1}, we focus on the two-dimensional case which was also analyzed in \cite{Milton:2017:SFI}. In particular, we will show that when the points $\Bx_1,\Bx_2, \ldots, \Bx_N$ are vertices of a convex polygon, our condition \eq{W9} is a generalization of the condition \eqref{I1} proved in \cite{Milton:2017:SFI}.  

\subsection{Mechanical interpretation of the extreme rays of the cone ${\CA_{\BX}}$ }

Given a set of points $\Bx_1,\Bx_2, \ldots, \Bx_N$ that are vertices numbered clockwise of a convex polygon, let us consider the following displacement field $\Bu_k^{(j,i)}$ defined by:
\beq \Bu_k^{(j,i)} = \begin{cases}
	-\BR_\perp(\Bx_k-\Bx_j),\quad \text{ for }k=j,j+1,\ldots, i-1, \\
	0\quad\text{ otherwise },
\end{cases}                   
\eeq{E2.1}
where 
\beq \BR_\perp=\left[\begin{array}{cc}
	0& 1\\-1 &0	
\end{array} \right] \eeq{Rper}
is the matrix for a 90$^\circ$ clockwise rotation and, if necessary, we identify $k$ with $k-N$. We call $\Bu_k^{(j,i)}$ a ``clam-shell'' displacement (see \fig{3}) as it corresponds to the infinitesimal rotation between two non-overlapping subpolygons the polygon of terminal points can be divided into: by keeping fixed one of the two subpolygons, the rotation of the other opens the clam. Given any points $\Bx_k$ and $\Bx_\ell$ on opposite sides of the clam (where $\Bu_k^{(j,i)}\ne 0$, while
$\Bu_\ell^{(j,i)}=0$) we have
\beq (\Bu_k^{(j,i)}-\Bu_\ell^{(j,i)})\cdot(\Bx_k-\Bx_\ell)=-[\BR_\perp(\Bx_k-\Bx_j)]\cdot(\Bx_k-\Bx_\ell)\geq 0,
\eeq{E2.3}
where the last inequality follows from the convexity of the polygon and the clockwise numbering of the points. Thus, this clam{-}shell movement is
an admissible displacement as it satisfies \eqref{W10}. This implies that $\BF$ satisfies the constraints \eqref{W9}, that is,
\beq 0\leq \sum_{k=1}^N \Bf_k\cdot\Bu_k^{(j,i)}=\sum_{k=j}^{i-1}(\Bx_k-\Bx_j)\cdot[\BR_\perp\Bf_k],
\eeq{E2.4}
which are precisely the same as the constraints \eq{I1} that characterize ${\CA_{\BX}}^*$, that is the set of all the loadings $\BF$ at $\BX$ which can be supported by a finite web. Thus, in this case of the $\Bx_i$ forming the vertices
of a convex polygon, the displacements $\BU$ correspond precisely
(up to an infinitesimal rigid body motion) to these ``clam-shell'' movements, and do not include any other movements.

\begin{figure}[!ht]
	\centering
	\includegraphics[width=0.6\textwidth]{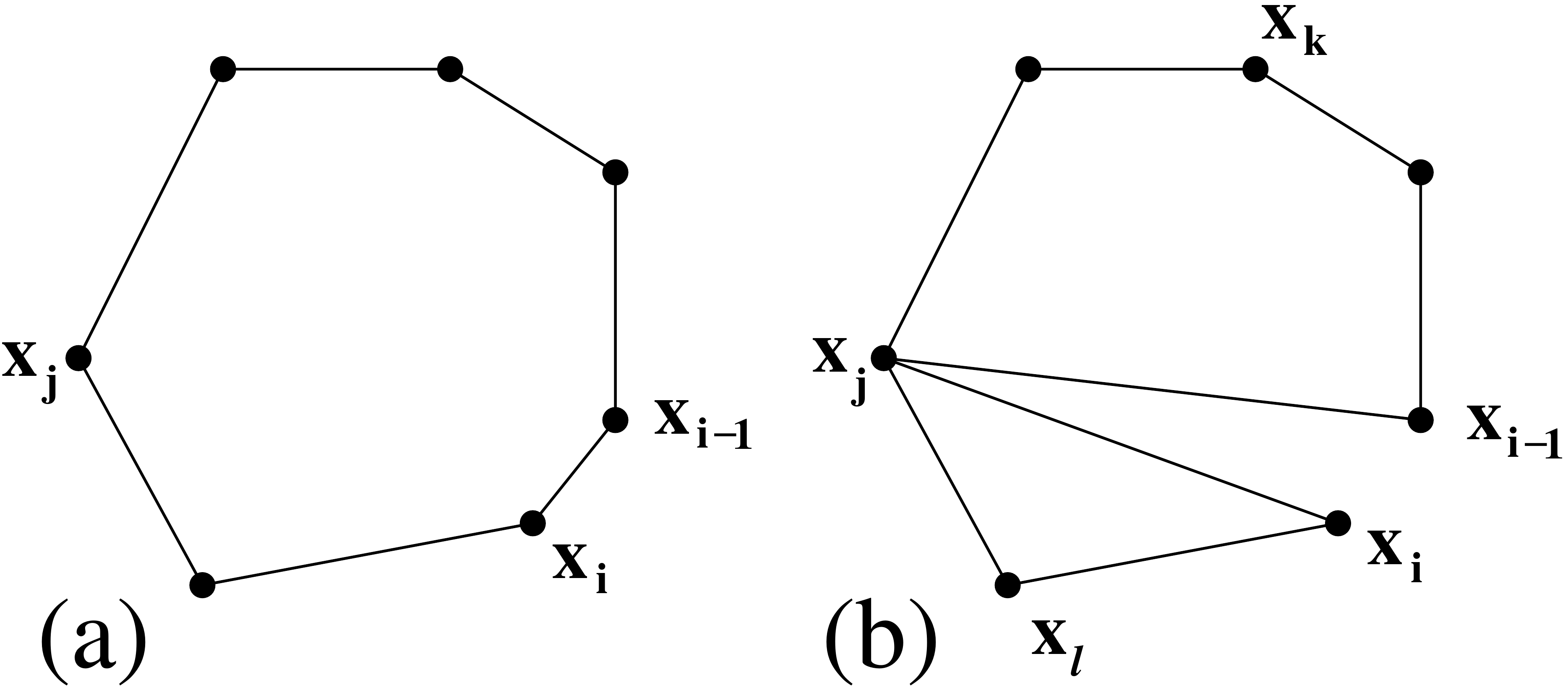}
	\caption{Given a set of points $\Bx_i$ that form the vertices of a convex polygon, as in (a), an extremal infinitesimal movement is {obtained} by breaking the polygon into two
		non-overlapping subpolygons connected at one vertex $\Bx_j$, as in (b). The ``clam-shell'' movement then consists of fixing one subpolygon, in this example the lower triangle,
		and infinitesimally rotating the other subpolygon anticlockwise about the point $\Bx_j$, so the ``clam'' opens slightly, thus moving any vertex $\Bx_k$ on the
		upper subpolygon away from any vertex $\Bx_\ell$ on the lower subpolygon.
	}
	\labfig{3}
\end{figure}

More generally, to check the criterion \eq{W9} it suffices to check it for those $\BU$ corresponding to the extreme rays $\BU^m$ of the cone ${\CA_{\BX}}$. These rays are perpendicular to the ``faces'' of the polar cone $-{\CA_{\BX}}^*$ (see Figure \fig{2}). 
\begin{figure}[!ht]
	\centering
	\includegraphics[width=0.6\textwidth]{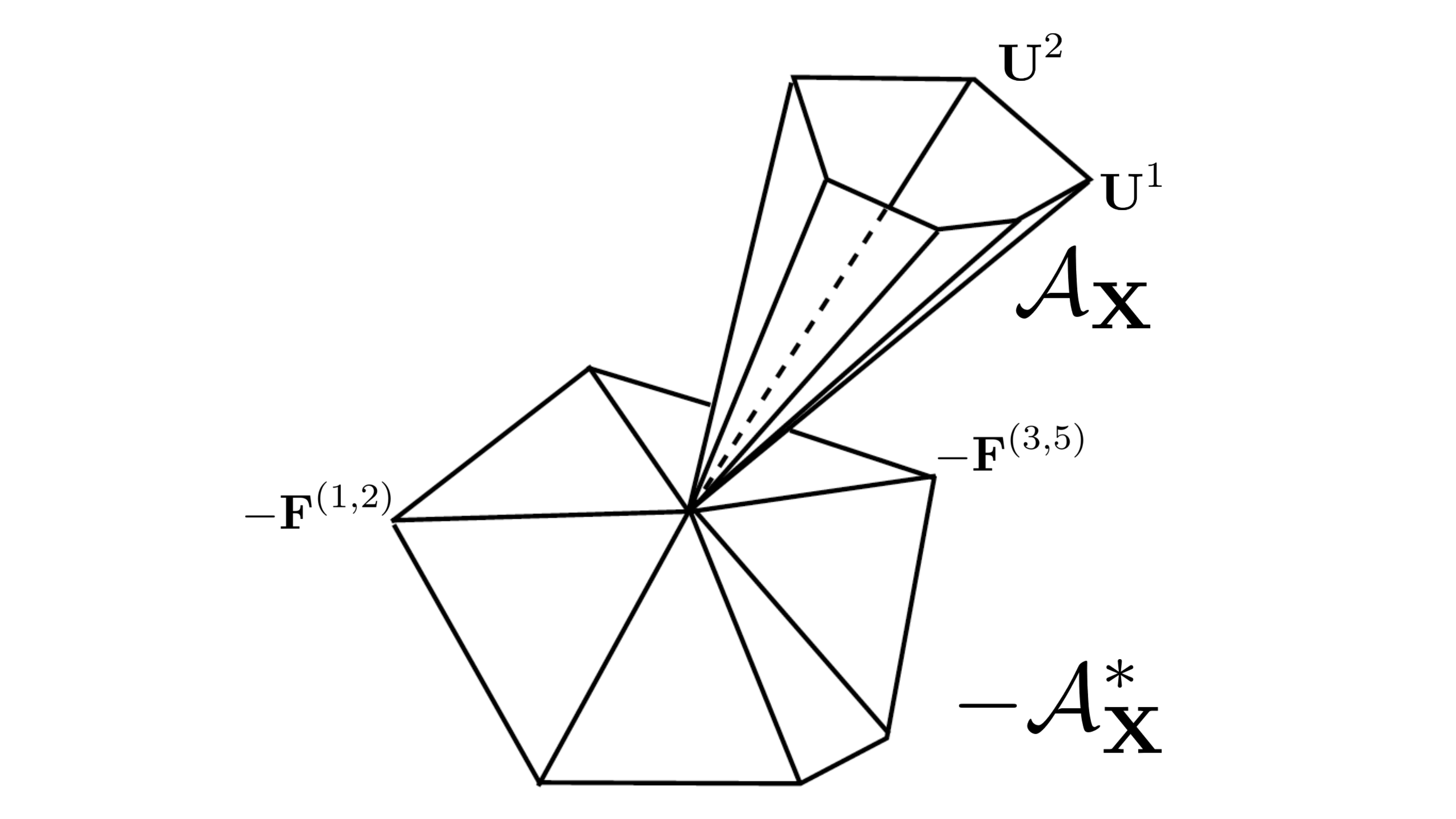}
	\caption{Schematic illustration of the cone $\CA_{\BX}$ and the polar cone $-\CA_{\BX}^*$, the negative of the dual cone $\CA_{\BX}^*$. Here the $\BF^{(i,j)}$ are the extreme rays of the dual cone $\CA_{\BX}^*$, while the
		$\BU^m$ are the extreme rays of the cone $\CA_{\BX}$.}
	\labfig{2}
\end{figure}
We use an integer $m=1,2,\ldots,M$ to index these rays. For any given
$m$ there exist associated loadings $\BF^m_h$, $h=1,2,\ldots,D-1$ all perpendicular to the extreme ray indexed by $m$. Here, for a given $m$, each value of $h$ signifies a
pair $(i,j)=(i(h,m),j(h,m))$ such that $\BF^{m}_h=\BF^{(i,j)}$, and linear combinations of the $\BF^m_h$, $h=1,2,\ldots,D-1$ with positive weights generate the ``face'' perpendicular to
the extreme ray of the cone ${\CA_{\BX}}$. Let $\BU^m=(\Bu_1^m,\Bu_2^m,\ldots,\Bu_N^m)$ be on this extreme ray. Then, we have
\beq \BF^{m}_h\cdot\BU^m=0,\text{ for }h=1,2,\ldots,D-1.
\eeq{RT5}
In particular, if $\BF^{m}_h=\BF^{(i(h,m),j(h,m))}$, then the orthogonality implies that
\beq (\Bu_{i(h,m)}^m-\Bu_{j(h,m)}^m)\cdot(\Bx_{i(h,m)}-\Bx_{j(h,m)})=0.
\eeq{RT6}
If we think of $\BU^m$ as corresponding to a displacement, then this restriction says that (within the infinitesimal displacements framework) there is no change in distance between
$\Bx_{i(h,m)}$ and $\Bx_{j(h,m)}$: the constraint is equivalent to only allowing those deformations compatible with rigid rods joining the pairs of points
$(\Bx_{i(h,m)},\Bx_{j(h,m)})$ for $h=1,2,\ldots,D-1$.
After eliminating the trivial infinitesimal rigid body motions (translations and rotations) from $\BU^m$, by requiring it to satisfy, for instance,
\beq\sum_{i=1}^N\Bu_i=\B0\,,\quad \sum_{i=1}^N\Bx_i\cdot\BA\Bu_i=\B0\eeq{AL3}
with $\BA$ any $d\times d$ antisymmetric matrix, there still must be one degree of freedom associated with the infinitesimal motion corresponding to $\BU^m$.
The equality $(\Bu_i-\Bu_j)\cdot(\Bx_i-\Bx_j)=0$ cannot hold for all $i\ne j$ as this then would correspond to a trivial overall (infinitesimal) rigid body motion,
which must be zero by \eqref{AL3}.  To fix the one-degree of freedom, we can impose the normalization condition that $\BU^m\cdot\BX=1$. Without loss
of generality this can be seen by noting that for any $\BU\in\CA_{\BX}$,
\beq 0\leq \sum_{i,j=1}^N(\Bu_i-\Bu_j)\cdot(\Bx_i-\Bx_j)=\sum_{i,j=1}^N\Bu_i\cdot\Bu_i+\Bu_j\Bu_j-2\sum_{i=1}^N\Bu_i\sum_{j=1}^N\Bu_j=2\BU\cdot\BX,
\eeq{RT7}
where, to get the last equality, we use \eq{AL3}.
Furthermore, $\BU\cdot\BX=0$ implies that $\BU$ is a rigid body motion, which we excluded. So $\BU\cdot\BX>0$, and by replacing $\BU$
with $\BU/(\BU\cdot\BX)$ we see that we can assume $\BU\cdot\BX=1$.

To conclude, we showed that if the terminal points are the vertices of a convex polygon, then all the extremal displacements correspond to clam-shell movements but notice that if there is at least one terminal point inside the convex hull of the terminal points, then besides clam-shell movements there are also other types of displacements, as shown in Figure \fig{Clam_shell_Arrowhead}.
\begin{figure}[h]
	\centering
	\begin{subfigure}{.45\textwidth}
		\centering
		\includegraphics[width=\textwidth]{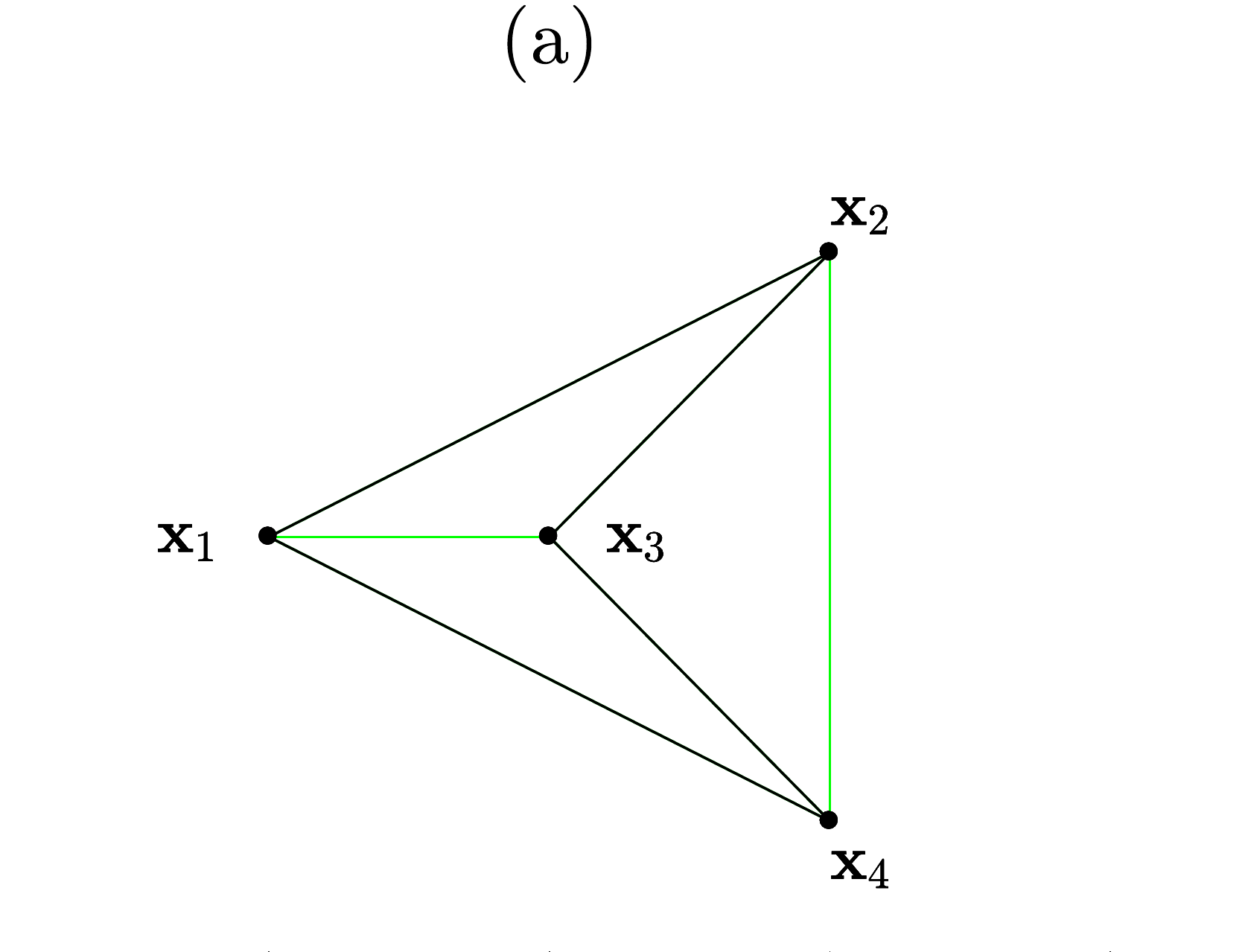}
	\end{subfigure}%
	\begin{subfigure}{.45\textwidth}
		\centering
		\includegraphics[width=\textwidth]{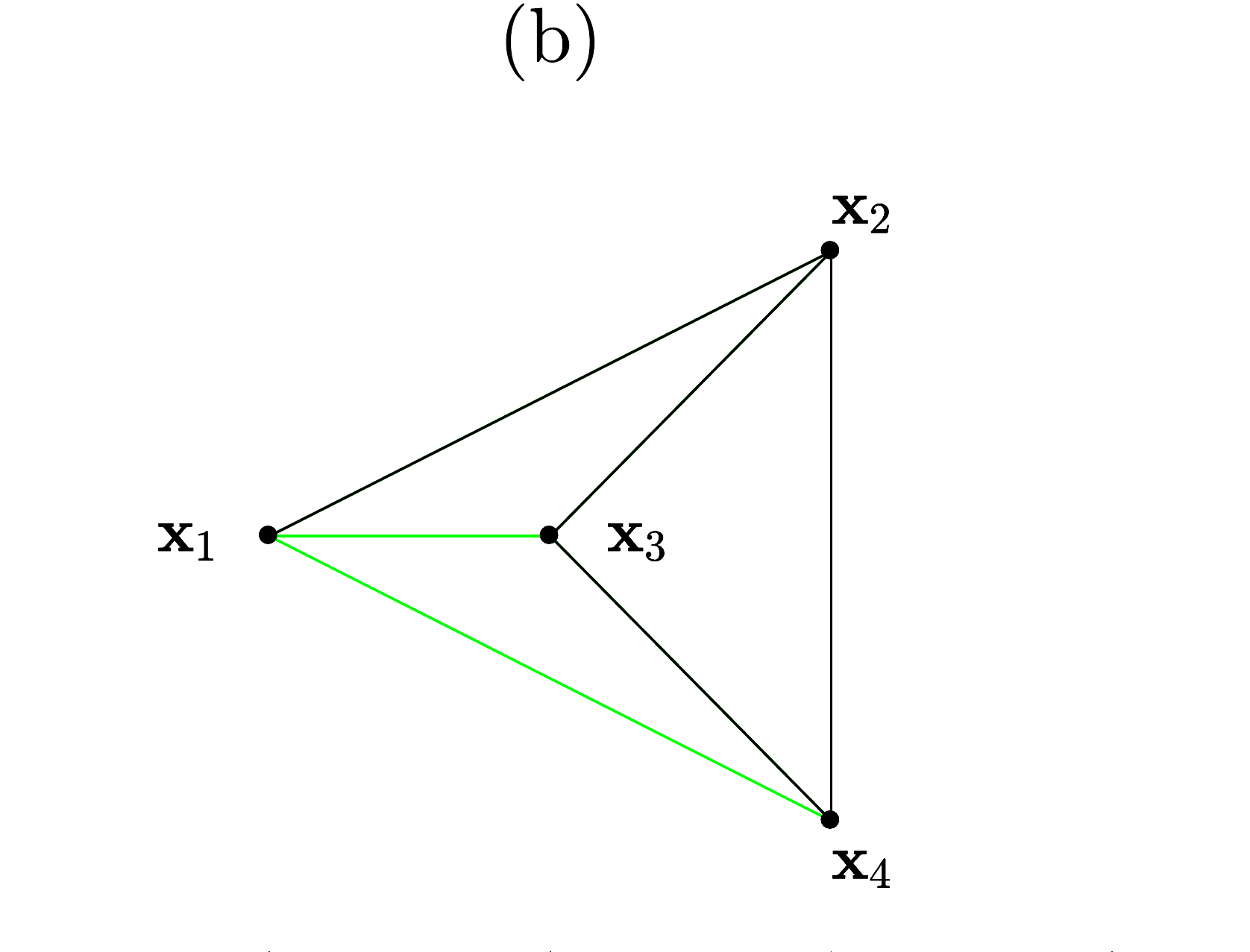}
	\end{subfigure}
	\caption{Consider the following 4 points which form an arrowhead in the plane: ${\bf x_1}=[0;0]$, ${\bf x_2}=[1;1]$, ${\bf x_3}=[0.5;0]$, and ${\bf x_4}=[1;-1]$. The displacements corresponding to the extreme rays of the cone ${\CA_{\BX}}$ for this special geometry can be divided into two groups: clam-shell movements and non-clam-shell movements. In this figure, the black lines represents the rigid wires and the green lines represents the deformable bonds. Clearly, (a) is an example of a non-clam-shell displacement, whereas (b) is an example of a clam-shell displacement (the rigid wire connecting the points ${\bf x_1}$ and ${\bf x_2}$ can rotate infinitesimally about the point ${\bf x_2}$ while the rigid triangle formed by ${\bf x_2}$, ${\bf x_3}$, and ${\bf x_4}$ is held fixed). Notice that the existence of non-clam-shell movements is due to the fact that the point ${\bf x_3}$ does not belong to the convex hull of the terminal points.}
	\labfig{Clam_shell_Arrowhead}
\end{figure}

\subsection{Simplifying the two-dimensional web}

In two dimensions, in any web, we define a loop to be any polygon such that its edges are the wires of the web. A minimal loop
is one with no other wires inside the polygon.
Any web with all pairs of terminal points interconnected, as in Figure \fig{1} (b), can then be replaced by an equivalent one with at most $P$ minimal loops
where $P$ is the number of points $\Bx_1,\ldots,\Bx_N$ that lie inside the convex hull of these $N$-points.
To show this we first place internal nodes
where any pair of wires $(\Bx_i,\Bx_j)$ and $(\Bx_r,\Bx_s)$ cross. Then take any minimal loop in the network. The vertices of this loop may include a terminal point $\Bx_j$
so long as the net force acting on the loop at  $\Bx_j$ (including $\Bf_j$ and the forces acting on $\Bx_j$ due to the tension in the other wires outside the loop)
points outside the loop. As the wires are all under tension, the loop then is necessarily convex
and exerts forces $-\Bf'_m$ at the nodes numbered clockwise around the loop. These forces necessarily satisfy \eq{I1}, and the loop can be replaced
by an open web. The number of minimal loops in the web is thus reduced by one. This procedure can be continued until there are at most $P$ minimal loops, and
each of these loops is non-convex and thus its vertices include at least one terminal node $\Bx_i$ where the associated force $\Bf_j$ points inside the loop. { Video1 (see Supplementary Material) shows an example of a web whose wires connect all the terminal points pairwise (initial frame), and how to replace each closed loop with an open web so that the final frame represents an equivalent web in which there is only one minimal loop due to the presence of one point, ${\bf x_4}$, inside the convex hull of the terminal points.}

\section{Channeling the stresses in a web}\label{Channelling}

In this Section we address the problem of designing wire webs that can support one and only one loading, up to a positive multiplicative factor. Note that the stress is not distributed in a unique way as there are many networks, i.e. stress patterns, that work. 
In a given wire web, this is possible if one can determine the stress in each wire in a unique way up to an overall proportionality constant: clearly this happens if at each internal node only four non-coplanar wires for the three-dimensional case (or three non-collinear wires for the two-dimensional case) meet and at most three non-coplanar wires meet at any terminal node for the three-dimensional case (or two non-collinear wires for the two-dimensional case). Here we provide a procedure to achieve such a goal so that at each internal node the coordination number is four for the three-dimensional case, or three for the two-dimensional case, while only one wire is connected to each terminal node.

Then, it is important to uniquely determine the loading that a web can support. Let us call $\CC^W_{\BX}$ the set of all the loadings $\BF$ that the
web $W$ can support at $\BX$. Clearly $\CC^W_{\BX}$ is a convex cone. Indeed, if the web supports the loadings $\BF^1$ and $\BF^2$ with admissible stresses $\ssigma^1$ and $\ssigma^2$, respectively, then for any $\Gl_1\geq 0$ and $\Gl_2\geq 0$ it also supports the loading $\BF=\Gl_1\BF^1+\Gl_2\BF^2$ with the associated stress $\Gl_1\ssigma^1+\Gl_2\ssigma^2$. Also, by definition, $\CC^W_{\BX}$ must be a subset of the admissible loading cone $\CA_{\BX}^*$. 

Here we address the converse question: given a convex cone $\CC\subset\CA_{\BX}^*$ can one find a web $W$ such that $\CC^W_{\BX}=\CC$?
We first focus on the case where $\CC$ is reduced to a single ray and then look for  what we call uniloadable webs (that is webs which support only one loading $\BF$, up to a multiplicative constant). If $\CC$ is a ray in the interior of $\CA_{\BX}^*$, then we can prove the existence of a uniloadable web for any ray $\CC$. If $\CC$ does not belong to the interior of $\CA_{\BX}^*$, then the existence of a uniloadable web is not guaranteed. Finally, we will answer the question in case $\CC$ is not simply a ray but a convex cone. 
Specifically, we will give a positive answer in the following asymptotic sense: one can find a sequence of finite webs $W_n$ such that $\CC^{W_n}_{\BX}$ approaches $\CC$
as $n\to\infty$. For two-dimensional webs where the points $\BX$ are the vertices of a convex polygon, a similar question was addressed by
Theorem 2 in \cite{Milton:2017:SFI}, and the proof given here is similar.

\subsection{Reducing the number of wires meeting at a point}
In two dimensions the  procedure for replacing a junction with $M>3$ wires by a localized web in which at most three wires meet is straightforward and described in Section 3 of {\cite{Milton:2017:SFI}.}
Briefly, and as illustrated in Figure \fig{5}, one finds the associated Airy stress function in the neighborhood of the junction. This is a convex cone with flat faces with the
discontinuity in slope at the edges corresponding to the tension in the wires ( Figure \fig{5}(b)). By cleaving the top of this cone, creating a polygonal face ({ Figure \fig{5}(c)}), one obtains an associated web ({Figure \fig{5}(d)})
supporting the same loading as the original junction ({Figure \fig{5}(a)}), but with at most three wires meeting at every junction.

\begin{figure}[!ht]
	\centering
	\includegraphics[width=0.6\textwidth]{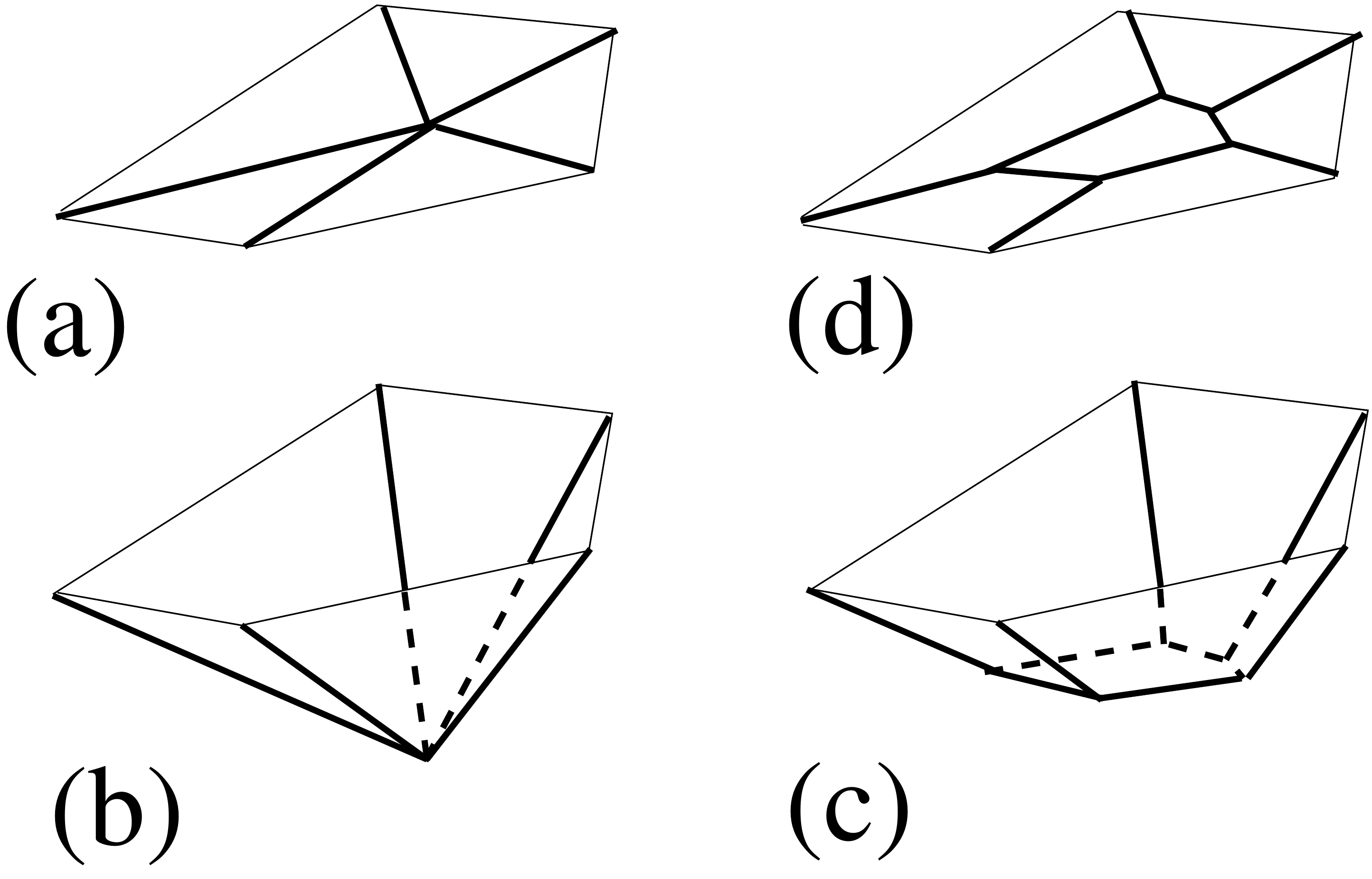}
	\caption{In two dimensions a junction of many wires at an internal node as in (a) can be replaced by a web in which at most three wires {meet} at every junction. First one determines an associated convex Airy
		function in the vicinity of the crossing, as in (b). Then one cleaves it by a plane as in (c), and the associated web, as in (d), then has  at most three wires meeting at every junction. }
	\labfig{5}
\end{figure}

In three dimensions the procedure we use for replacing a junction with $M>4$ wires is more complicated. The steps are illustrated in Figure \fig{6}, where we begin in (a) with a junction
where $M=6$ wires meet. First we pick 4 of the wires (those marked in blue in Figure \fig{6}(a)), and points $\Bx_1$, $\Bx_2$, $\Bx_3$ and $\Bx_4$
on the four wires such that the tetrahedron $T$ formed by $\Bx_1$, $\Bx_2$, $\Bx_3$ and $\Bx_4$ encloses the junction, which without loss of generality can be taken to be at the origin $\Bx_0=0$
({this} requires that the four wires be chosen so that they do not all lie on one side of any plane through the origin{: balance} of forces at the origin ensures that at least one choice
of such four wires exists). The tensions in these four wires generally do not balance. However, consider the ``tensegrity network'' consisting of rods from the origin to
the points $\Bx_1$, $\Bx_2$, $\Bx_3$ and $\Bx_4$ under compression, balanced by wires along the edges of the tetrahedron $T$ that are under tension. Example \ref{Ex1} { (see Section\ref{Introduction})} gives the
explicit solution for the compressive forces in the rods and the tensions in the wires in this ``tensegrity network''. We next superimpose this ``tensegrity network'' on our
junction with the tensions in the ``tensegrity network'' scaled so after superposition the tension near the junction in one wire cancels, while the tension in the other wires remain
nonnegative, as sketched in Figure \fig{6}(b). We thus obtain a web under tension where the number of wires joined to the origin is now $M-1$ or less. However, we typically have also created junctions, at
some of the points $\Bx_1$, $\Bx_2$, $\Bx_3$ and $\Bx_4$ where $5$ wires meet, like those in Figure \fig{6}(c).
These junctions are rather special in that one wire goes straight through the junction (but
typically has different tensions on opposite sides of the junction). These junctions are then locally replaced by networks in which at most four wires meet as illustrated in
Figure \fig{6}(d). Lemma \ref{5wires} below guarantees this can be done. The last step is to successively repeat the argument until the junction at the origin has at most 4 wires.
\begin{figure}[!ht]
	\centering
	\includegraphics[width=0.6\textwidth]{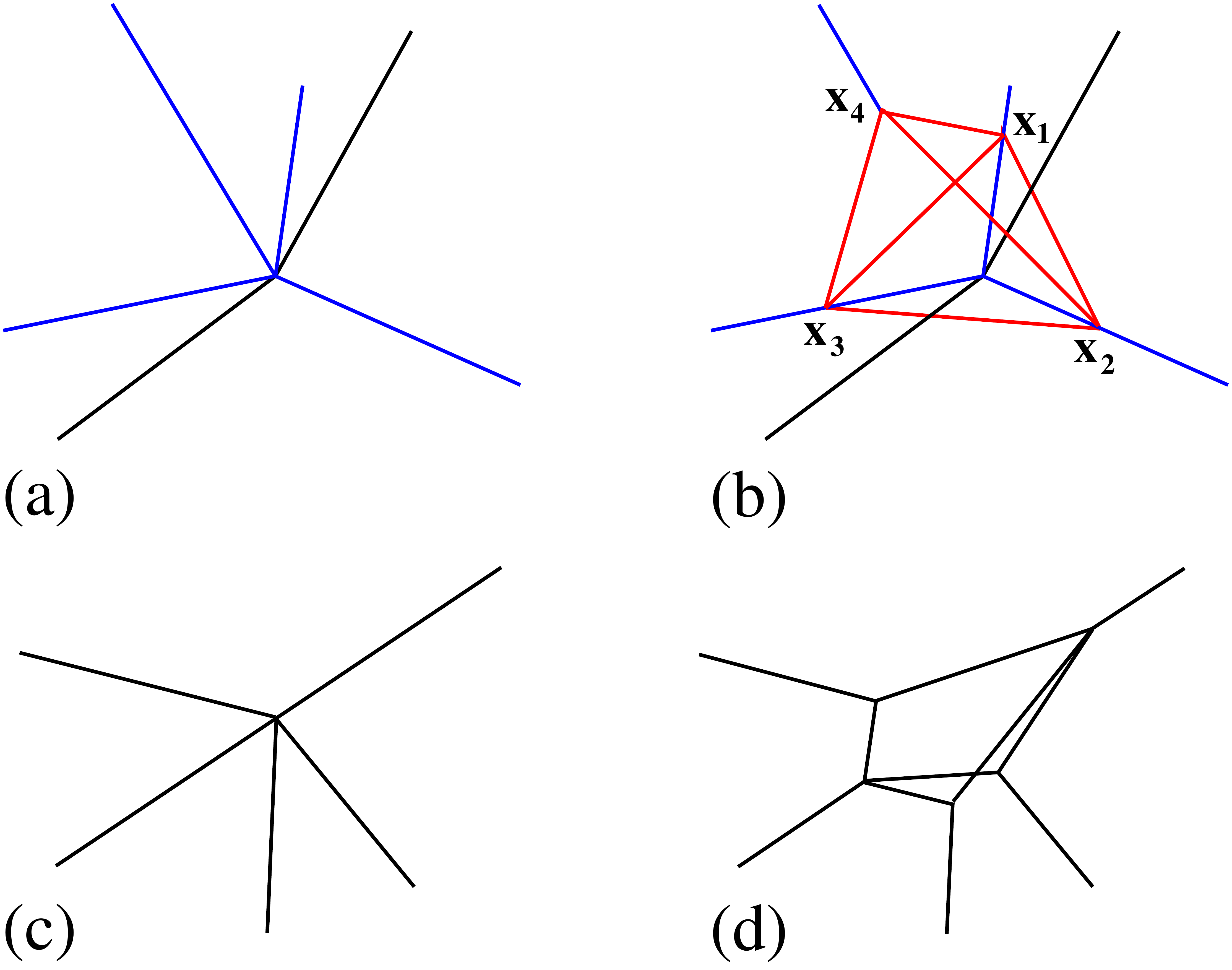}
	\caption{Steps in the replacement of a junction of many wires under given balanced tensions, with a network localized around the junction such that at most 4 wires meet at any junction in the new network, and the network still supports the same tensions in the wires meeting the network.}
	\labfig{6}
\end{figure}
\begin{lemma}\label{5wires}{\bf The five wires problem}
	\newline
	Consider five wires with tensions $T_i>0$, directions $\Bv_i$ and joining at the origin $\Bx={\bf 0}$. Assume that the three first directions are independent while $\Bv_5=- \Bv_4$. Then we can replace these wires by a web in tension such that at each of its nodes, no more than 4 wires are joining.
\end{lemma}
\begin{proof}
	We can assume without loss of generality that $T_5=\alpha T_4$ so that balance of forces implies that $T_1 \Bv_1 +T_2 \Bv_2+T_3 \Bv_3+(1-\alpha )T_4 \Bv_4=0$.
	Set $t>0$ and $s> \frac {\alpha}{3}$. Set also $\Bx_1=t T_1 \Bv_1$, $\Bx_2=t T_2 \Bv_2$, $\Bx_3=t T_3 \Bv_3$, $\Bx_4=t\, s T_4 \Bv_4$ and $\Bx_5=-t\, r T_4 \Bv_4$ where $r:=\frac{s}{3s-\alpha}$. As the real parameter $t$ can be arbitrarily chosen provided it is small enough, then we can avoid creating new nodes where a node already exists or where a wire lies. The parts of the wires which lie between the points $\Bx_i$ and the origin are replaced by six wires $[\Bx_i,\Bx_4]$, $[\Bx_i,\Bx_5]$ ($i\in\{1,2,3\}$). When these wires have respective (positive) tensions
	\beq T_{i4}=\frac{r}{t(r+s)}{\|\Bx_i-\Bx_4 \|},\qquad T_{i5}=\frac{s}{t(r+s)}{\|\Bx_i-\Bx_5 \|}, \eeq{5W.1}
	the web is in equilibrium. Indeed at each node $\Bx_i$ for $i\in\{1,2,3\}$ we have
	\beq
	T_{i4}  \frac{\Bx_4-\Bx_i}{\|\Bx_i-\Bx_4 \|} + T_{i5} \frac{\Bx_5-\Bx_i}{\|\Bx_i-\Bx_5 \|} +  T_i \Bv_i
	= \frac{rs-sr}{t(r+s)}t T_4 \Bv_4  + \frac{r+s}{t(r+s)}(-t T_i \Bv_i) +  T_i \Bv_i =0.
	\eeq{5W.2}
	Moreover, at  nodes $\Bx_4$ and $\Bx_5$ we have
	\beq \sum_{i=1}^3 \Big( T_{i4}  \frac{\Bx_i-\Bx_4}{\|\Bx_i-\Bx_4 \|}\Big)  +T_4 \Bv_4
	= \frac{1}{r+s}  \Big( r\alpha +s- 3rs\Big) T_4 \Bv_4 =0.
	\eeq{5W.3}
	\beqa \sum_{i=1}^3 \Big( T_{i5}  \frac{\Bx_i-\Bx_5}{\|\Bx_i-\Bx_5 \|}\Big)  +T_5 \Bv_5
= \frac{1}{r+s}  \Big( - s-\alpha r  + 3rs\Big) T_4 \Bv_4 =0.
	\eeqa{5W.4}
\end{proof}


\subsection{Uniloadable webs}\label{Uniloadable_webs}

Given a ray $\CC\subset\CA_{\BX}^*$, we want to determine whether there exists a uniloadable web which can support the corresponding loading $\lambda\BF$, with $\lambda\geq 0$. In case $\CC$ does not belong to the interior of the cone $\CA_{\BX}^*$, then the existence of a uniloadable web is not guaranteed. When $\BF$ belongs to the interior of $\CA_{\BX}^*$, then we prove the existence of a uniloadable web supporting such a loading (up to a multiplicative constant).

\subsubsection{Stuck loadings}\label{secnonstuck}

Let us start with the following remark : let $\BF=(\Bf_1,\dots,\Bf_N)$ belong to  $\CA_{\BX}^*$ with $\BX=(\Bx_1,\dots,\Bx_N)$. Then, it also belongs, for any  $t \geq 0$, to the admissible loading cone $\CA_{\BX+t\BF}^*$ of the shifted points $(\Bx_1+t\Bf_1,\dots, \Bx_N+t\Bf_N)$. Indeed, taking a web  $W$  which supports $\BF$ at $\BX$ and adding to $W$ all the wires $[\Bx_i,\widetilde \Bx_i]$ leads clearly to a web supporting $\BF$ at $\BX+t\BF$. When $t<0$ things are less clear. We say that $\BF$ at $\BX$ is an {\it unstuck  loading} if there exists $\varepsilon>0$ such that
$$\exists \varepsilon >0,\  \forall t<\varepsilon,\ \BF\in \CA_{\BX-t\BF}^*,$$
otherwise we say that  $\BF$ is a {\it stuck loading}. A particular case of stuck loadings, referred to as {\it completely stuck loadings}, is  when there exist some $k\in\{1,\dots,N\}$ and $\varepsilon >0$ such that,
$$\forall\, 0<t<\varepsilon, \  \BF\not \in \CA_{\widetilde \BX}^* \ \text{  with }\  \widetilde \BX=(\Bx_1,\dots,\Bx_{k-1},\Bx_k-t\Bf_k,\Bx_{k+1},\dots,\Bx_N).$$

The stuck or completely stuck conditions can only occur when the loading $\BF$ is a ray on the boundary of the cone of admissible loadings $\CA_{\BX}^*$, as we will prove in the next subsection that any $\BF$ in the interior of $\CA_{\BX}^*$ is unstuck. For a better insight, let us consider two examples of completely stuck loadings.

\begin{Aexa}\label{Example3}
	Consider forces {$\Bf_1=[-1;0]$, $\Bf_2=[3/4;1]$, $\Bf_3=[3/4;-1]$, and $\Bf_4=[-1/2;0]$ at the four points considered in Figure \fig{Clam_shell_Arrowhead}, that is, ${\bf x_1}=[0;0]$, ${\bf x_2}=[1;1]$, ${\bf x_3}=[1;-1]$, and ${\bf x_4}=[1/2;0]$}. They are supported by the web in Figure \fig{4}(a).
	
	\begin{figure}[!ht]
		\centering
		\includegraphics[width=0.7\textwidth]{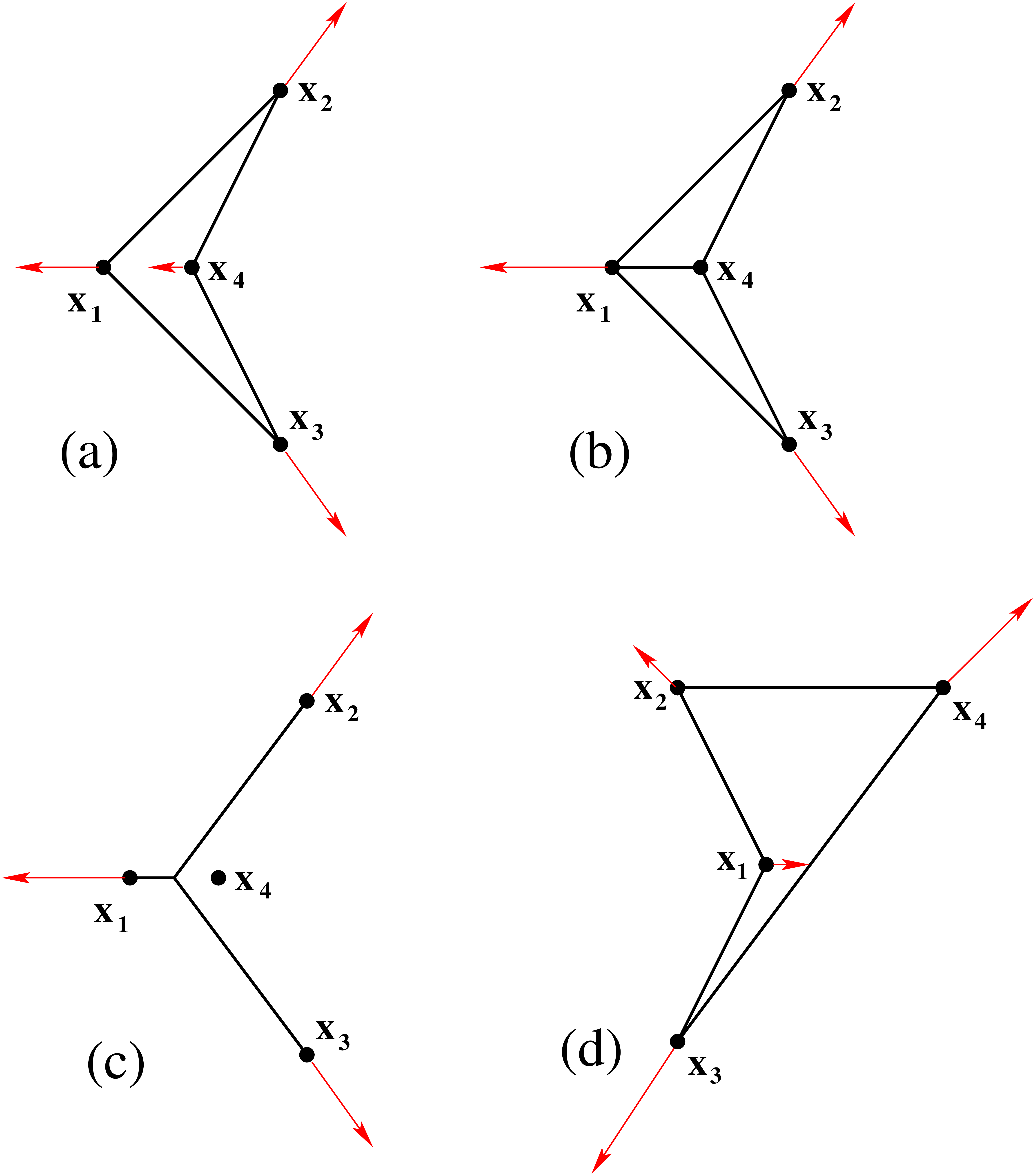}
		\caption{Examples of webs with "completely stuck loadings" in two dimensions. We argue in Example \ref{Example3} that (a) is the unique web that supports the given forces at the terminal nodes. A wire can be attached to the node
			$\Bx_4$, as in (b), in order to uniquely define the associated Airy stress function, up to the addition of an affine function. This Airy stress function then lies above the triangular pyramidal Airy stress function associated
			with the net forces applied  at the three points $\Bx_1$, $\Bx_2$, and $\Bx_3$, as in (c). At the same time it must have a discontinuity in slope across the line between $\Bx_1$ and $\Bx_4$ corresponding to the tension in the
			wire joining those two points. This then uniquely determines the Airy stress function (modulo the addition of an affine function) and thus uniquely determines the web. Figure (d), as discussed in Example \ref{Example4},
			provides another example of a web with a "completely stuck loading" where the web is uniquely determined once the loading is specified.}
		\labfig{4}
	\end{figure}
	
	Our objective is to show that this is the only web supporting these forces, thus implying that the loading is
	``completely stuck''. First move the force {$\Bf_4=[-1/2;0]$} back to the origin, by attaching a wire joining the origin to $\Bx_4$, as in Figure \fig{4}(b). Then
	the net force acting on the origin is $[-3/2;0]$ and is balanced by the forces $\Bf_3$ and $\Bf_4$. The open web in Figure \fig{4}(c) supports these three forces, and the associated Airy
	function (up to addition of an affine function, see also \cite{Fraternali:2014:OTC}) is
	\beq \Gf_L(\Bx)=\max\{-3|x_2|/4,1/4-x_1 \},
	\eeq{SW1}
	where $x_1$ and $x_2$ are the coordinates of 
$\Bx$.
	
	Now suppose we have any web supporting the four forces $\Bf_1$, $\Bf_2$, $\Bf_3$, and $\Bf_4$. Necessarily this web will be confined within the convex hull of the three points
	$\Bx_1,\Bx_2,\Bx_3$. To define the associated Airy stress potential, in say all of $\R^2$, we need to move the four forces
	to infinity by attaching {infinite} wires to the four points in the direction of infinity (the wire attached to $\Bx_1$ overlaps the wire attached to $\Bx_4$ left of the origin).
	The existence of these wires implies a discontinuity of slope of the Airy
	stress function across them, matching the tension in the wire. Now consider the Airy stress function in the vicinity of the point $\Bx_4$. In particular take the
	tangent plane at a point $\Bx_0$ that approaches $\Bx_4$ from the left remaining infinitesimally above the wire that joins $\Bx_4$ to the origin. Similarly
	take the tangent plane at a point $\Bx_0'$ that approaches $\Bx_4$ from the left remaining infinitesimally below the wire that joins $\Bx_4$ to the origin. The maximum of these
	two tangent planes is the valley function $\Gf_V$ that takes the form
	\beq \Gf_V(\Bx)=-|x_2|/4+ax_1+bx_2+c.
	\eeq{SW2}
	As the web is confined to the convex hull of $\Bx_1,\Bx_2,\Bx_3$, the Airy stress function outside this convex hull can be taken to be $\Gf_L$ (modulo an affine function
	that can be set to zero without loss of generality). Convexity of the Airy stress function then implies the inequalities $ \Gf_L(\Bx_i)\geq\Gf_V(\Bx_i)$ for $i=1,2,3$
	and $\Gf_L(\Bx_4)\leq\Gf_V(\Bx_4)$. Elementary calculations then show that these inequalities allow one no freedom in the choice of $a$, $b$ and $c$ and one necessarily
	has $a=-1/2$, $b=c=0$. By convexity, the Airy stress function of any web supporting the four forces must be above
	\beq \Gf(\Bx)=\max\{\Gf_L(\Bx),\Gf_V(\Bx)\}=\max\{-3|x_2|/4,(1/4)-x_1,-|x_2|/4-x_1/2\}, \eeq{SW3}
	and must coincide with it at the four points $\Bx_1,\ldots,\Bx_4$. But the polyhedral nature of $\Gf(\Bx)$ means that any other candidate convex Airy stress function
	must cleave it {in} the vicinity of $\Bx_4$, which is forbidden. Thus (modulo the addition of an affine function) $\Gf(\Bx)$ given by \eq{SW3} is the unique possibility
	for the  Airy stress function, and the web in Figure \fig{4}(a), is the only web that can support the four forces at the four points.
\end{Aexa}

\begin{Aexa}\label{Example4}
	A second example of a web with completely stuck loading is shown in Figure \fig{4}(d). Forces {$\Bf_1=[2;0]$, $\Bf_2=[-2;2]$, $\Bf_3=[-4;-6]$, and
		$\Bf_4=[4;4]$ are applied at the points $\Bx_1=[0;0]$, $\Bx_2=[-1/2;1]$, $\Bx_3=[-1/2;1]$, and $\Bx_4=[1;1]$}. The unique web that supports
	them is that drawn in Figure \fig{4}(d), and the associated Airy stress potential (modulo addition of an affine function) is
	\beq \Gf(\Bx)=\max\{2x_1,-2x_2+1,-|x_2|-(4x_1-1-3x_2)^+\}, \eeq{SW4}
	where $q^+=\max\{0,q\}$. The proof proceeds similarly to the first example. The completely stuck nature of the webs in our two examples has been verified numerically.
\end{Aexa}

In two dimensions when the terminal points $\BX$ are at the vertices of a convex polygon, the existence of an open web supporting (the assumed admissible) loading $\BF$ implies
that such webs are never "completely stuck" with respect to the loading $\BF$. One may wonder: is a similar result true in three dimensions? The numerical example of {Figure \fig{cube}}
shows, to the contrary, that there exist terminal points $\BX$ at the vertices of a convex polyhedron, and an admissible loading $\BF$ such that the associated web is "completely stuck".
This example was found by starting with $N=8$ terminal nodes at the vertices of a cube, taking an admissible loading $\BF$ supported at these points, then moving the terminal points backwards
so that for each $j=1,2,\ldots, 8$, $\Bx_j$ is replaced  by $\Bx_j'=\Bx_j-\e_j\Bf_j$, where the $\e_j\geq 0$ are increased until the loading $\BF$ at the terminals $\BX'=(\Bx'_1,\Bx_2',\ldots,\Bx'_8)$
is completely stuck, while keeping the terminals $\BX'$ as vertices of a convex polyhedron.

\begin{figure}[h!]
	\centering
	\begin{subfigure}{.5\textwidth}
		\centering
		\includegraphics[width=\textwidth]{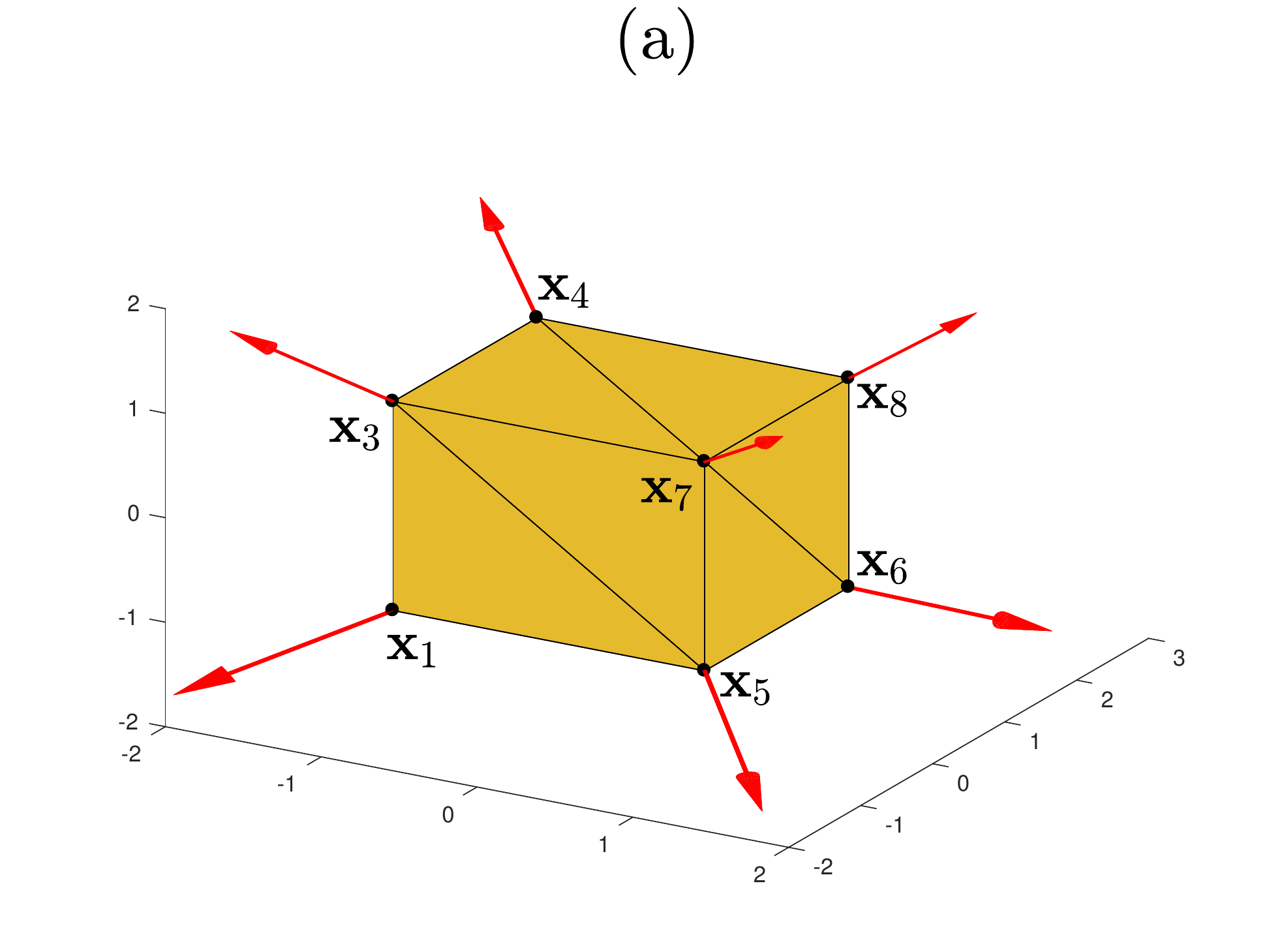}
	\end{subfigure}%
	\begin{subfigure}{.5\textwidth}
		\centering
		\includegraphics[width=\textwidth]{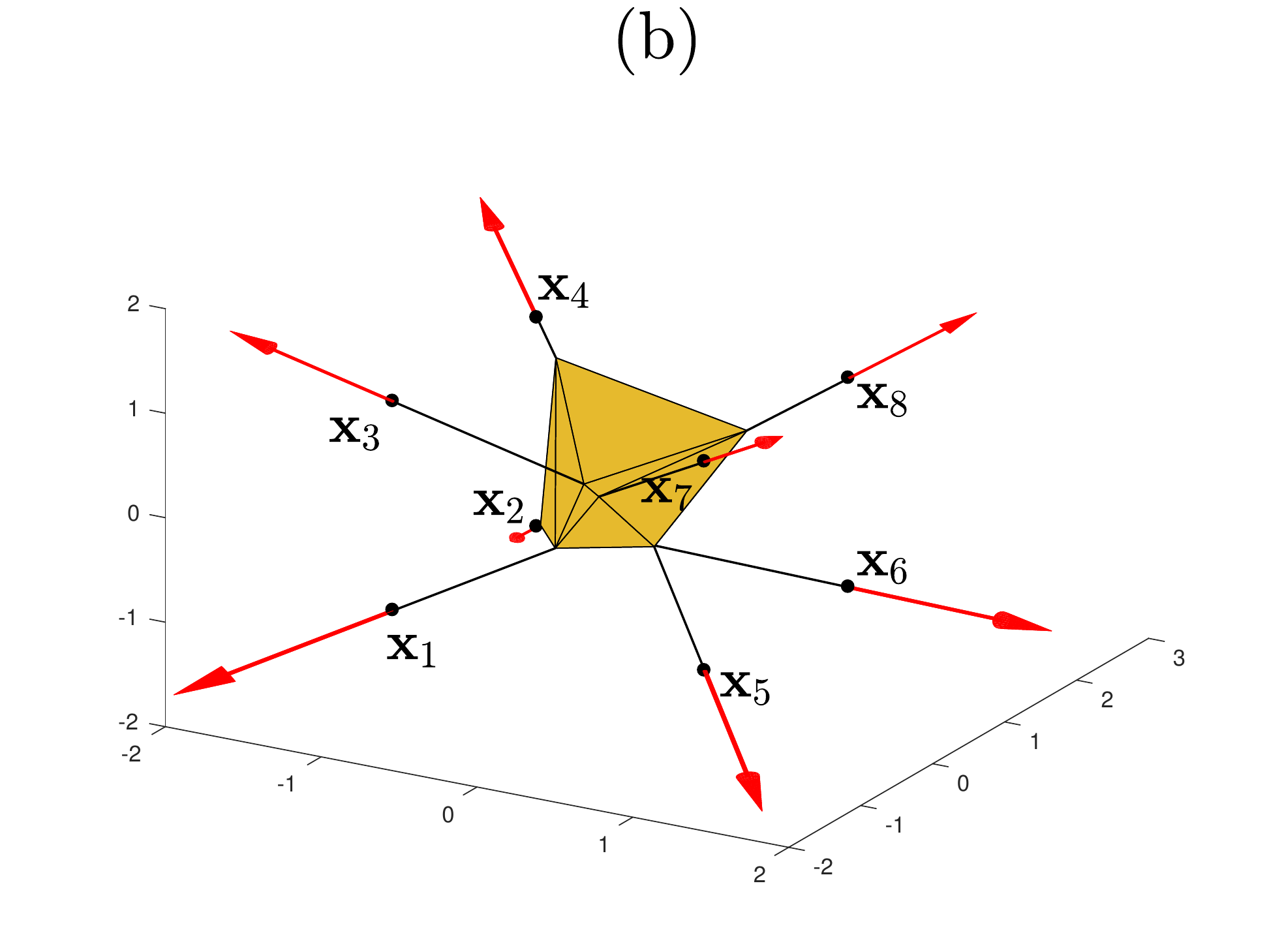}
	\end{subfigure}
	\caption{We start by applying the forces ${\Bf}_i$, given in Table \ref{table:1}, at the 8 points shown in (a), whose coordinates are given by ${\Bx}_i$ in Table \ref{table:1}. The points are initially chosen to be the vertices of a cube. Now, move the terminal points backwards
		so that for each $i=1,2,\ldots, 8$, $\Bx_i$ is replaced  by $\Bx_i'=\Bx_i-\e_i\Bf_i$, where the $\e_i\geq 0$ are increased until the loadings ${\Bf}_i$ at the terminals $\Bx_i'$, provided in Table \ref{table:1},
		is completely stuck, while keeping $\Bx_i'$ as vertices of a convex polyhedron. The diagonal lines along the faces in figure (a) are numerical artifacts.}
	\labfig{cube}
\end{figure}
\begin{table}[h!]
	\centering
	\begin{tabular}{c |c |c |c}
		\hline
		$i$ & ${\Bx}_i$ & ${\Bx}_i'$ & ${\Bf}_i$ \\
		\hline
		1 & [-1;-1;-1] & [-0.27473;-0.31827;-0.46693] & [-0.98151;-0.92259;-0.72140] \\
		2 & [-1;1;-1] & [-0.94086;0.91949;-0.93330] & [-0.46129;0.62796;-0.52022]\\
		3 & [-1;-1;1] & [-0.12837;-0.23757;0.15689] & [-0.74581;-0.65237;0.72140] \\
		4 & [-1;1;1] & [-0.76226;0.74617;0.78972] & [-0.72140;0.77022;0.63807] \\
		5 & [1;-1;-1] & [0.28777;-0.16187;-0.35258] & [0.80474;-0.94700;-0.73151] \\
		6 & [1;1;-1] & [0.27806;-0.12953;-0.35765] & [0.75592;1.1827;-0.67259] \\
		7 & [1;-1;1] & [-0.01231;-0.28017;0.08708] & [0.74581;-0.53033;0.67259] \\
		8 & [1;1;1] & [0.51981;0.62495;0.51177] & [0.60355;0.47140;0.61366] \\
		\hline
	\end{tabular}
	\caption{Components of the forces ${\Bf}_i$ applied at the 8 points shown in Figure \fig{cube} whose coordinates in the original configuration are given by ${\Bx}_i$ (Figure \fig{cube}(a)), and in the final configuration by ${\Bx}_i'$ (Figure \fig{cube}(b)).}
	\label{table:1}
\end{table}


\subsubsection{Unstuck loadings}

The aim is to prove that, for a given loading $\BF$ belonging to the interior of the one of admissible loadings ${\cal A}^*_{\BX}$, there exists a uniloadable web supporting only $\BF$ (up to a multiplicative constant). First, we need to prove that the cone ${\cal A}^*_{\BX}$ does not have an empty interior (Lemma \ref{nonempty}). Then, we need to prove that if we perturb slightly the positions of the terminal points, the loading is not stuck (Lemma \ref{lemnonstuck}) and that there exists a connected web with all the wires under tension which supports $\BF$. Finally, we can prove that such a web is a uniloadable web (Theorem \ref{mainthm2}).

To prove that the cone ${\cal A}^*_{\BX}$ does not have an empty interior
we need to introduce the space $\CR_{\BX}$ of infinitesimal rigid motions on $\BX=(\Bx_1,\dots,\Bx_N)$ and its orthogonal $\CB_{\BX}$, called the space of {\it balanced loadings} on $\BX$, defined by 
\beqa \CR_{\BX} &:= & \{ \BU=(\Bu_1,\dots, \Bu_N)\in (\R^d)^N:\exists a \in \R^d, \ \exists \BA \text{ antisymmetric such that}, \forall i\,\,\,  \Bu_i=\Ba+ \BA\Bx_i\}, \nonum
\CB_{\BX}& := & \{ \BF=(\Bf_1,\dots, \Bf_N)\in (\R^d)^N:\ \sum_{i=1}^N \Bf_i=0,\ \sum_{i=1}^N (\Bf_i\otimes \Bx_i - \Bx_i\otimes \Bf_i)=0\}.
\eeqa{NE.1}
There is no loss of generality to assume that $\sum_{i=1}^N \Bx_i=0$ so that $\BX$ also belongs to $\CB_{\BX}$.

Recall that ${\cal A}^*_{\BX}$, the cone of admissible forces $\BF=(\Bf_1,\dots,\Bf_N)$ at the points $\BX=(\Bx_1,\dots,\Bx_N)$, is the dual cone of  ${\cal A}_{\BX}:=\{U\in \CB_{\BX}:\ \forall (i,j),\ (\Bu_i-\Bu_j)\cdot (\Bx_i-\Bx_j)\geq 0\}$ (see Theorem \ref{mainthm1}). The set $\CB_{\BX}$ is a subspace of $(\R^d)^N$ (with codimension $d(d+1)/2$). The notion of interior we use is relative to this subspace.

\begin{lemma} \label{nonempty}
	The cone ${\cal A}^*_{\BX}$ is a subset of $\CB_{\BX}$ with non empty interior.
\end{lemma}
\begin{proof}
	We have already noticed, as a consequence of Theorem \ref{Thsupport},  that there is no loss of generality in assuming that the points $(\Bx_1,\dots \Bx_{N})$  span the space $\R^d$. There is no loss of generality either in assuming that it is the first points $(\Bx_1,\dots \Bx_{d+1})$ which span it.
	
	Assume now, by contradiction, that the set  ${\cal A}^*_{\BX}$ has empty interior. As it is convex, this means that it is contained in a lower dimension subspace: there exists $\BU=(\Bu_1,\dots,\Bu_N)\not=0 \in \CB_{\BX}$ such that $\BF\cdot\BU =0$ for all $\BF\in {\cal A}^*_{\BX}$. As $(\Bx_1,\dots \Bx_{d+1})$ span $\R^d$, and there exists a unique affine function $\Bu$ such that $\Bu(\Bx_i)=\Bu_i$ for $1\leq i \leq d+1$.

	Clearly, for any $(i,j)\in\{1,\dots,N\}$, the particular loading $\BF^{(i,j)}=(\Bf_1,\dots,\Bf_N)${\red ,} defined by $\Bf_i=\Bx_i-\Bx_j$, $\Bf_j=\Bx_j-\Bx_i$ and $\Bf_k=0$ whenever $k\not=i$ and $k\not=j$, belongs to
	${\cal A}^*_{\BX}$ (indeed a simple wire linking $\Bx_i$ to $\Bx_j$ is an admissible web for this particular loading). Hence $\BU$ satisfies $(\Bu_i-\Bu_j)\cdot(\Bx_i-\Bx_j)=0$ for any pair $(i,j)$. This condition applied to all pairs with $i\leq d+1$ and $j\leq d+1$ implies that $\Bu$ is a rigid motion. The same condition applied to all pairs $(i,j)$ with  $i\leq d+1$ and $j> d+1$ implies that
	$\Bu_j=\Bu(\Bx_j)$ too. Then $\BU$ is a non vanishing rigid motion and this contradicts the definition of $\CB_{\BX}$.\
\end{proof}
Let us now prove that all loadings in the interior of   ${\cal A}^*_{\BX}$ are unstuck. More precisely,
\begin{lemma}\label{lemnonstuck}
	Let $\BF$ be in the interior of ${\cal A}^*_{\BX}$. Then, for $\e>0$ small enough, $\BF$ also belongs to the interior of  ${\cal A}^*_{\BX-\e \BF}$.
\end{lemma}
\begin{proof} We first remark that, for any $\e\in\R$, $\BF$ belongs to $\CB_{\BX-\e \BF}$. Indeed, for any antisymmetric matrix $\BA$ we have $\sum_{i=1}^N \Bf_i\cdot (\BA (\Bx_i-\e \Bf_i))=-\e \sum_{i=1}^N \Bf_i\cdot (\BA \Bf_i)= 0$.
	
	\smallskip
	Let now $\BU$ be any vector in ${\cal A}_{\BX-\e \BF}$ with $\|\BU\|=1$. For any pair $(i,j)$, it fulfills
	\beq (\Bu_i-\Bu_j)\cdot ((\Bx_i-\e \Bf_i)-(\Bx_j-\e \Bf_j))\geq 0, \eeq{NS.1}
	which implies
	\beq (\Bu_i-\Bu_j)\cdot (\Bx_i-\Bx_j)\geq \e (\Bu_i-\Bu_j)\cdot (\Bf_i-\Bf_j)\geq -4 \e  \|\BF\|. \eeq{NS.2}
	Let $\delta:=\min_{i\not =j} |\Bx_i-\Bx_j|$ be the smallest distance between the points $\Bx_i$ and $\gamma:=\frac {4 \|\BF\|}{\delta^2}$. The vector $\BW:=\BU+\e\gamma \BX$ satisfies
	\beq (\Bw^i-\Bw^j)\cdot (\Bx_i-\Bx_j)= (\Bu_i-\Bu_j)\cdot (\Bx_i-\Bx_j)+\e \gamma (\Bx_i-\Bx_j)\cdot (\Bx_i-\Bx_j) \geq 0. \eeq{NS.3}
	Therefore its projection $\overline \BW$ on $\CB_{\BX}$, which satisfies the same inequalities, belongs to ${\cal A}_{\BX}$.
	
	The cone ${\cal A}^*_{\BX}$ is the set of all $\BF$ satisfying $\forall \BU\in {\cal A}_{\BX}$, $\BF\cdot\BU  \geq 0$. As $\BF$ belongs to the interior  of ${\cal A}^*_{\BX}$, the function $\BV\to \BF\cdot\BV$ must be strictly positive on the compact intersection of ${\cal A}_{\BX}$ with the unit sphere. Let us call $\alpha>0$ its minimum. By homogeneity, we have for any $\BV$ in  ${\cal A}_{\BX}$, $\BF\cdot\BV \geq  \alpha \|\BV\|$.
	When applied to $\overline \BW$, this inequality gives
	\beq \BF\cdot\BW= \BF\cdot\BW = \BF\cdot\BU +\e \gamma \BF\cdot\BX \geq  \alpha \|\overline \BW\|=\alpha\|\overline{\BU+\e\gamma \BX}\|\geq \alpha( 1-\e\gamma \|\BX\|). \eeq{NS.4}
	Hence we have
	\beq  \BF\cdot\BU\geq \alpha( 1-\e\gamma \|\BX\|)-\e \gamma \BF\cdot\BX\geq  \alpha -\e\gamma \|\BX\| (\alpha+\|\BF\|). \eeq{NS.5}
	For $\e$ smaller than $\frac {\alpha}{2 \gamma\|X\| (\alpha+\|\BF\|)}$, we get  $\BF\cdot\BU\geq \frac \alpha 2  \|\BU\|$. This inequality, valid for any  $\BU\in{\cal A}_{\BX-\e \BF}$ with $\|\BU\|=1$, remains true on the whole cone ${\cal A}_{\BX-\e \BF}$ by homogeneity. So the lemma is proven.
\end{proof}

\begin{lemma}\label{lemconnected}
	
	If $\BF$ is in the interior of ${\cal A}^*_{\BX}$ then there exists a connected web supporting $\BF$ at $\BX$ with  a strictly positive stress state.
\end{lemma}
\begin{proof}
	
	We first check that there exists a web $W$ supporting $\BF$ at $\BX=(\Bx_1,\dots,\Bx_N)$ such that all wires with non vanishing tension make a connected set.
	The set $\mathcal D$ of loadings $\BF$ for which there exists a set $I$ strictly included in $\{1,\dots,N\}$ such that  $\sum_{i\in I} \Bf_i=\B0$ is a union of subspaces (of codimension $d$) of the space $\mathcal B_\BX$ of balanced loading.
	We want to avoid the case where $\BF$ is supported by two or more disconnected webs and clearly this may only happen if $\BF\in {\mathcal D}$. In that case we use the following trick: as the interior of $\mathcal D$ is empty we can find $\BG\in {\mathcal B}_\BX$ such that
	$\BF+\BG$ and $\BF-\BG$ belong to $\CA_{\BX}^*\setminus {\mathcal D}$. Let ${W}^+$ and ${W}^-$ be two webs supporting respectively   $\BF+\BG$ and $\BF-\BG$ and let $\ssigma^+$ and $\ssigma^-$ be the associated stress measures. The measure  $\ssigma:=(\ssigma^++\ssigma^-)/2$ satisfies
	$$\Div\ssigma+\CF =\Div\ssigma^+/2+\Div\ssigma^-/2+\CF= -(\CF + \CG)/2 -(\CF - \CG)/2 +\CF= \B0.$$
	Its support is the union of the supports of $\ssigma^+$ and $\ssigma^-$ which both are connected sets containing all the terminal nodes where the applied forces are non vanishing.
	Thus we get a finite connected web  $W$  with a strictly positive stress state $\sigma$  that supports the loading $\BF=(\Bf_1,\dots,\Bf_N)$ at points $\BX=(\Bx_1,\dots,\Bx_N)$. 
\end{proof}

\begin{corollary}\label{cornonstuck}
	If $\BF$ is in the interior of ${\cal A}^*_{\BX}$ then there exists is a connected web supporting $\BF$ at $\BX$ with  a strictly positive stress state and only one wire joining each terminal node.
\end{corollary}

\begin{proof}
	Lemma \ref{lemnonstuck} provides an $\e>0$ small enough for $\BF$ to belong to the interior of  ${\cal A}^*_{\BX-\e \BF}$. Lemma \ref{lemconnected} then  provides a connected web $W$ supporting $\BF$ at $\BX-\e \BF$ with a strictly positive stress state $\sigma$. Adding to $W$, for those $i\in\{1,\dots,N\}$ such that $\Bf_i\not=0$, the wires $[\Bx_i,\Bx_i+\e \Bf_i]$ and fixing the tension to $\|\Bf_i\|$ in each of these wires, makes a new web supporting $\BF$ at $\BX$ with a strictly positive stress state. Clearly, as $W$ is connected,  the new one is connected too. 
\end{proof}

We can finally state the theorem regarding the existence of a uniloadable web for an unstuck loading.
\begin{theorem}\label{mainthm2}{\bf Existence of uniloadable webs}
	\newline
	For any $\BF$ in the interior of the admissible loading cone $\CA_{\BX}^*$, there exists a finite web $W$ such that  $\CC^W_{\BX}=\{\Gl\BF: \  \Gl\geq 0\}$.
\end{theorem}
\begin{proof}
	Corollary \ref{cornonstuck} states that there exists a finite connected web  $W$  with a strictly positive stress state $\sigma$  that supports the loading $\BF=(\Bf_1,\dots,\Bf_N)$ at points $\BX=(\Bx_1,\dots,\Bx_N)$ and  such that only one wire is attached to each terminal node.
	This web is then modified (according to Lemma \ref{lemnonstuck}) into a new web $\widetilde W$ having the extra property that all the internal nodes have in two  dimensions at most
	three coplanar wires meeting it, and in three dimensions at most four noncoplanar wires, or three coplanar wires, meeting it.
	This ensures that  $\widetilde W$ supports only the loadings $\Gl\BF$ for $\Gl\geq 0$. Indeed, as  $\BF$ is an admissible loading  for $\widetilde W$ at points $\BX$, the unique wire attached to $\Bx_1$ has direction $\Bf_1$. Let $\widetilde \BF=(\widetilde\Bf_1,\dots,\widetilde\Bf_N)$ be another admissible loading  for $\widetilde W$ at $\BX$. Balance of forces at $\Bx_1$ imposes $\widetilde\Bf_1=\Gl \Bf_1$. As the wire attached to $\Bx_1$ is under tension for both loadings, we have $\Gl\ge 0$. Now at each node of the web, once the (positive) tension in one of the joining wires is fixed, the balance of forces fixes the tension in all other joining wires. As the web is connected, from node to node, the tensions of all wires are fixed :  $\widetilde \BF$ is uniquely determined and clearly $\widetilde \BF=\Gl \BF$.
\end{proof}

\subsection{Possible loading cones}

Here we seek answer to the question: given a convex cone $\CC\subset\CA_{\BX}^*$ can one find a web $W$ such that $\CC^W_{\BX}=\CC$?

To proceed, as sketched in Figure 4 of \cite{Milton:2017:SFI},
one approximates the convex cone $\CC$ by a cone, that we will denote with  $\CC^{W_j}_{\BX}$, having a finite number $j$ of extreme rays $\BF^{(j)m}$, $m=1,2,\ldots,j$, each
strictly in the interior of $\CC$ and hence strictly in the interior of the admissible loading cone $\CA_{\BX}^*$. As we are free to make arbitrary small perturbations of the extreme rays  $\BF^{(j)m}$, we can assume that these are chosen so that at each terminal $i$, $\Bf^{(j)m}_i$ is not collinear with $\Bf^{(j)\ell}_i$ for any $m,\ell$ with $\ell\ne m$.
Associated with $\BF^{(j)m}$ is then a uniloadable web $W^{(j)m}$ supporting, and only supporting, at the terminals $\BX$,
the loadings $\Gl\BF^{(j)m}$, $\Gl\geq 0$. By adjusting the positions of the interior nodes of the different uniloadable webs $W^{(j)m}$ now indexed by $j=1,2,\ldots,m$
we can ensure that the webs $W^{(j)m}$ do not have overlapping interior nodes, nor overlapping collinear wires.
By superimposing these uniloadable webs, for  $m=1,2,\ldots,j$ one obtains the web $W_j$ having the desired loading cone. It may happen that some wire pairs cross when we superimpose the webs,
generating an additional interior node at the crossing point. This is still okay, as balance of forces at the crossing point ensures that the tension in the wire remains the same on opposite sides of the crossing point. If more than two wires cross at the same point, then we can perturb the uniloadable webs to avoid this.

Taking the limit $j\to\infty$ allows us to approximate $\CC$ arbitrarily closely, so that $\CC^{W_j}_{\BX}\to\CC$.
%

\section{Conclusions}

{In this paper we provide full answer to the question as to whether, given a set of forces applied to specific points, called terminal points, there exists a web that supports such forces with all the wires under tension. Specifically, we provide a necessary and sufficient condition on the loading forces which guarantees the existence of such a web, see Theorem \ref{mainthm1}. Such a condition corresponds to a finite dimensional linear programming problem: if this has solution, then a web exists which is formed by the wires connecting pairwise the terminal points.  Conversely, any web under tension supporting the given loading can be replaced by the web provided by the linear programming problem.
	
	The conditions related to the linear programming problem are inequalities expressed in terms of the displacements of the terminal points: they form a cone in the displacement space and the edges of the cone correspond to those displacements that satisfy the conditions as equalities. In case the terminal points form the vertices of a convex polygon, these extreme edges correspond precisely (up to an infinitesimal rigid body motion) to clam-shell movements and do not include any other movements. By clam-shell movements we denote a displacement field that arises when one breaks the convex polygon connecting the terminal points into two non-overlapping subpolygons connected at one vertex and fixes one subpolygon and infinitesimally rotates the other one, so the clam opens
	slightly. In the case there is at least one terminal point that
	lies inside the convex hull of the terminal points, the extreme rays of the cone of admissible displacements correspond to types of
	displacements that are not simply clam-shell movements. 
	
	In practical situations one would like to have uniloadable webs, that is webs that support only one loading and all positive multiples of it: such webs allow one to channel stresses in desired ways and the superposition of them allows one to get a desired convex loading cone. To construct a uniloadable web in
	two-dimensions, one has to replace each closed minimal loop, that is any polygon formed by the intersection of
	wires that cannot be divided into subpolygons, by an open web (not containing any closed loop). This is always possible if the terminal points are positioned at
	the vertices of a convex polygon. If that is not the case, then there
	will still be minimal closed loops in a number equal to that of the terminal points which lie inside the convex hull of the
	terminal points.

	In general, to construct uniloadable webs, one has to reduce the number of wires meeting at either a terminal point or at an internal node. The first step is to modify the web so that only one wire is attached to each terminal point. We proved that this is possible only when the given loading lies inside the cone formed by the inequality conditions associated with the dual linear programming problem. If, instead, the loading corresponds to one of the extreme rays of such a cone, then we have that for some webs this modification is impossible and we say that such webs have stuck loadings. We provided two examples of webs of this type in two dimensions, see Examples \ref{Example3} and \ref{Example4}. Since in two dimensions we know that a web with terminal points that are vertices of a convex polygon can always be replaced by an open web, we have that such webs are never completely stuck: indeed, in both Example \ref{Example3} and \ref{Example4} the terminal points are not vertices of a convex polygon. This does not hold in the three-dimensional case for which one can find a web with stuck loadings which have the terminal points forming a convex polyhedron (see Figure \fig{cube}). On the other hand, if the loading belongs to the interior of the cone formed by the inequalities regarding the loadings in the dual linear programming problem, then we provide a general procedure to reduce the number of wires meeting at each internal node where initially 5 or more wires meet.
	
\vskip6pt

\enlargethispage{20pt}

\section*{Acknowledgments}
Ornella Mattei and Graeme Milton are grateful to the National Science Foundation for support through {DMS-1814854}.
Graeme Milton is grateful to the Universit{\'e} de Toulon for hosting his two week visit there in September 2017 where this work was initiated.



\begin{thebibliography}{99}

\bibitem{Ahmad:2013:NLA}
M.~S. Ahmad Abad, A. Shooshtari, V. Esmaeili and A.~N. Riabi. Nonlinear Analysis of cable structures under general loadings. {\em Finite Elements in Analysis and Design}, 73:11--19, 2013. https://doi.org/10.1016/j.finel.2013.05.002.

\bibitem{Ahmadizadeh:2013:TDG}
M. Ahmadizadeh.
 Three-dimensional geometrically nonlinear analysis of slack cable structures.
 {\em Computers\& Structures},
128:160--169, 2013. https://doi.org/10.1016/j.compstruc.2013.06.005.

\bibitem{Andreu:2006:AND}
A. Andreu, L. Gil, P. Roca.
 A new deformable catenary element for the analysis of cable net structures.
 {\em Computers\& Structures},
84:1882--1890, 2006. https://doi.org/10.1016/j.compstruc.2006.08.021.

\bibitem{Bouchitte:2008:MTL}
G. Bouchitt{\'e}, W. Gangbo, and P. Seppecher.
 Michell trusses and lines of principal action.
 {\em Mathematical Models and Methods in Applied Sciences},
18(9):1571--1603, 2008. https://hal.archives-ouvertes.fr/hal-00528973.

\bibitem{Camar:2003:DCS}
M.~Camar-Eddine and P.~Seppecher.
 Determination of the closure of the set of elasticity functionals.
 {\em Archive for Rational Mechanics and Analysis}, 170(3):211--245,
2003. https://doi.org/10.1007/s00205-003-0272-7.

\bibitem{Crusells:2017:AMF}
M.~Crusells-Girona, F.~C. Filippou and R.~L. Taylor.
 A mixed formulation for nonlinear analysis of cable structures.
 {\em Computers \& Structures}, 186:50--61,
2017. https://doi.org/10.1016/j.compstruc.2017.03.011.

\bibitem{Ekeland:1999:CAV}
{I}. Ekeland and {R.} T{\'e}mam.
 Convex analysis and variational problems,
(English ed.), 
 {\em Classics in Applied Mathematics}, vol.~28, Society for Industrial
and Applied Mathematics (SIAM), Philadelphia, PA, 1999 (translated from French). 


\bibitem{Farkas:1902:TDE}
J. Farkas.
 Theorie der einfachen Ungleichungen. 
 {\em Journal f{\:u}r die reine und angewandte Mathematik}, 124:1--27, 1902.

\bibitem{Fraternali:2014:OTC}
F. Fraternali and G. Carpentieri.
 On the correspondence between 2D force networks and polyhedral stress functions. 
 {\em International Journal of Space Structures}, 29:145--159, 2014. https://doi.org/10.1260/0266-3511.29.3.145.

\bibitem{Vasquez:2011:CCS}
F. {Guevara Vasquez}, G.~W. Milton, and D. Onofrei.
 Complete characterization and synthesis of the response function of
elastodynamic networks.
 {\em Journal of Elasticity}, 102(1):31--54, 2011. https://doi.org/10.1007/s10659-010-9260-y.

\bibitem{Kirszbraun:1934:UDZ}
M.D., Kirszbraun.  {\"U}ber die zusammenziehende und Lipschitzsche Transformationen.  {\em Fundamenta Mathematicae}, 22: 77--108.


\bibitem{Michell:1904:LEM}
A.G.M Michell.  The limits of economy of material in frame structures. {\em The London, Edinburgh, and Dublin Philosophical Magazine and Journal of Science}, 8: 589--597.


\bibitem{Milton:2017:SFI}
G.~W. Milton.
 The set of forces that ideal trusses, or wire webs, under tension can
support.
 {\em International Journal of Solids and Structures}, 128(Supplement
C):272--281, 2017. https://doi.org/10.1016/j.ijsolstr.2017.08.035.

\bibitem{Milton:1995:WET}
G.~W. Milton and A.~V. Cherkaev.
 Which elasticity tensors are realizable?
 {\em ASME Journal of Engineering Materials and Technology},
117(4):483--493, 1995. https://doi.org/10.1115/1.2804743.

\bibitem{Norris:2008:ACT}
A.N. Norris. Acoustic cloaking theory. {\em Proceedings of the Royal Society A} 464: 2411--2434, 2008. https://doi.org/10.1098/rspa.2008.0076.

\bibitem{Such:2009:AAB}
M. Such, J. Jimenez-Octavio, A. Carnicero and O. Lopez-Garcia.
 An approach based on the catenary equation to deal with static analysis of three dimensional cable structures.
 {\em Engineering Structures},
31:2162--2170, 2009. https://doi.org/10.1016/j.engstruct.2009.03.018.

\bibitem{Thai:2011:NSD}
H. Thai and S. Kim.
 Nonlinear static and dynamic analysis of cable structures.
 {\em Finite Elements in Analysis and Design},
47:237--246, 2011. https://doi.org/10.1016/j.finel.2010.10.005.

\bibitem{Yang:2007:GNA}
Y. B. Yang and J. Y. Tsay.
 Geometrical nonlinear analysis of cable structures  with a two-node cable element by generalized displacement control method.
 {\em International Journal of Structural Stability and Dynamics},
7:571--588, 2007. https://doi.org/10.1142/S0219455407002435.


\end{thebibliography}
\end{document}